\documentclass{article}
\def\full{1} 
\pdfoutput=1 

\usepackage[final]{neurips_2020} 

\usepackage[utf8]{inputenc} 
\usepackage[T1]{fontenc}    
\usepackage{hyperref}       
\usepackage{url}            
\usepackage{booktabs}       
\usepackage{amsfonts}       
\usepackage{nicefrac}       
\usepackage{microtype}      
\usepackage{comment}
\usepackage{grffile}

\usepackage{enumitem}
\usepackage{hyperref} 
\usepackage{mathtools,verbatim}
\usepackage{color}
\usepackage{xcolor}
\definecolor{darkgreen}{rgb}{0.1,0.8,0.1}
\usepackage{amssymb}
\usepackage{times}
\usepackage{wrapfig}
\usepackage{amsmath} 
\usepackage[utf8]{inputenc}
\usepackage{adjustbox}
\usepackage{caption}
\usepackage[margin=10pt,font=small,labelfont=bf]{caption}
\usepackage{placeins}
\usepackage{amsthm}
\usepackage{graphicx}
\usepackage{multirow}
\usepackage{clrscode}
\usepackage{subcaption}

\newtheorem{theorem}{Theorem}
\newtheorem{lemma}{Lemma}
\newtheorem{proposition}{Proposition}
\newtheorem{definition}{Definition}
\newtheorem{corollary}{Corollary}

\newtheorem{fact}{Fact}[section]
\newtheorem{example}{Example}[section]

\newtheorem{assumption}{Assumption}[section]

\newcommand{\normf}{\dagger}

\newcommand{\transpose}{\mathrm{T}}

\newcommand{\reals}{\mathbf{R}}

\newcommand{\trace}{\mathrm{Tr}}

\newcommand{\Kld}{\mathrm{KL}}
\newcommand{\Rank}{\mathrm{rank}}

\newcommand{\calE}{\mathcal{E}}
\newcommand{\calR}{\mathcal{R}}

\newcommand{\bbB}{\mathbb{B}}

\newcommand{\calp}{\mathcal{P}}

\newcommand{\lt}{\underline{t}}
\newcommand{\ut}{\overline{t}}

\newcommand{\bbC}{\mathbb{C}}
\newcommand{\bbD}{\mathbb{D}}
\newcommand{\bbN}{\mathbb{N}}
\newcommand{\calN}{\mathcal{N}}
\newcommand{\bbS}{\mathbb{S}}
\newcommand{\bbU}{\mathbb{U}}

\newcommand{\mzero}{\mathbf{0}}
\newcommand{\mc}{\mathbf{c}}

\newcommand{\mP}{\mathbf{P}}
\newcommand{\mr}{\mathbf{r}}
\newcommand{\mmu}{\mathbf{u}}
\newcommand{\mmv}{\mathbf{v}}
\newcommand{\mmw}{\mathbf{w}}

\newcommand{\mx}{\mathbf{x}}
\newcommand{\my}{\mathbf{y}}
\newcommand{\mz}{\mathbf{z}}

\newcommand{\mX}{\mathbf{X}}
\newcommand{\mY}{\mathbf{Y}}
\newcommand{\mZ}{\mathbf{Z}}

\newcommand{\calI}{{\mathcal I}}

\newcommand{\calP}{{\mathcal P}}

\newcommand{\spans}{\mathrm{Span}}

\newcommand{\vx}{\mathbf{x}}

\newcommand{\vZ}{\mathbf{Z}}

\newcommand{\E}{\mathrm{I\! E}}

\providecommand{\mypara}[1]{\smallskip\noindent\emph{#1} }
\providecommand{\myparab}[1]{\smallskip\noindent\textbf{#1} }

\definecolor{myred}{HTML}{880000}
\definecolor{mygreen}{HTML}{008800}
\definecolor{myblue}{HTML}{000088}
\definecolor{linkblue}{HTML}{0000BB}
\usepackage{hyperref}
\hypersetup{
  colorlinks,
  linkcolor={myred},
  citecolor={myblue},
  urlcolor={linkblue}
}

\PassOptionsToPackage{numbers}{natbib}

\usepackage{neurips_2020}

\usepackage[symbol]{footmisc}

\title{Adaptive Reduced Rank Regression}

\newcommand\thankssymb[1]{\textsuperscript{{#1}}}

\author{
	Qiong Wu\thankssymb{}\thanks{\thankssymb{} Correspondence to: Qiong Wu <qwu05@email.wm.edu>. } \\
 William \& Mary \\
	\And
	Felix M. F. Wong \\ 
	Independent Researcher\thankssymb{}\thanks{\thankssymb{} Currently at Google.}
	\And
	Yanhua Li \\
	Worcester Polytechnic Institute \And
	Zhenming Liu \\
	 William \& Mary \\
	\And 
	Varun Kanade \\
	{University of Oxford}
}

\begin{document}
\maketitle

\ifnum\full=1
\vspace{-8mm}
\tableofcontents
\newpage
\fi

\begin{abstract}
We study the low rank regression problem $\my = M\mx + \epsilon$, where $\mx$ and $\my$ are $d_1$ and $d_2$ dimensional vectors respectively. We consider the extreme high-dimensional setting where the number of observations $n$ is less than $d_1 + d_2$. Existing algorithms are designed for settings where $n$ is typically as large as $\Rank(M)(d_1+d_2)$. This work provides an efficient algorithm which only involves two SVD, and establishes statistical guarantees on its performance. The algorithm decouples the problem by first estimating the precision matrix of the features, and then solving the matrix denoising problem. To complement the upper bound, we introduce new techniques for establishing lower bounds on the performance of any algorithm for this problem. Our preliminary experiments confirm that our algorithm often out-performs existing baselines, and is always at least competitive.
\end{abstract}
\section{Introduction}
We consider the regression problem $\my = M \mx + \epsilon$ in the high
dimensional setting, where $\mx \in \reals^{d_1}$ is the vector of features,
$\my \in \reals^{d_2}$ is a vector of responses, $M \in \reals^{d_2 \times
d_1}$ are the learnable parameters, and $\epsilon \sim N(0,
\sigma^2_{\epsilon}I_{d_2 \times d_2})$ is a noise term. 
High-dimensional setting refers to the case where the number of observations
$n$ is insufficient for recovery and  
hence regularization for estimation is
necessary~\cite{koltchinskii2011nuclear,negahban2011estimation,chen2013reduced}.
This high-dimensional model is widely used in practice, such as identifying
biomarkers~\cite{zhu2017low}, understanding risks associated with various
diseases~\cite{frank2015dietary,batis2016using}, image
recognition~\cite{rahim2017multi,fan2017hyperspectral}, forecasting equity
returns in financial
markets~\cite{polk2006cross,stock2011dynamic,medeiros2012estimating,bender2013foundations},
and analyzing social networks~\cite{wu2017generalized,ralescu2011spectral}. 

We consider the ``large feature size'' setting, in which the number of features $d_1$ 
is excessively large and can be even larger than the number of observations $n$. This setting frequently arises in practice because it is often straightforward to perform feature-engineering and produce a large number of potentially useful features in many machine learning problems. For example, in a typical equity forecasting model, $n$ is around 3,000 (i.e., using 10 years of market data), whereas the number of potentially relevant features can be in the order of thousands~\cite{polk2006cross,gu2018empirical,kelly2019characteristics,chen2019deep}. In predicting the popularity of a user in an online social network, $n$ is in the order of hundreds (each day is an observation and a typical dataset contains less than three years of data) whereas the feature size can easily be more than 10k~\cite{ravikumar2010high,bamman2014gender,sinha2013predicting}.

Existing low-rank regularization techniques (e.g.,
\ifnum\full=0
~\cite{anderson1951estimating, izenman1975reduced, koltchinskii2011nuclear,negahban2011estimation,ma2014adaptive}
\fi\ifnum\full=1
~\cite{koltchinskii2011nuclear,negahban2011estimation,ma2014adaptive}
\fi) are not optimized for the large feature size setting. These results assume that either the features possess the so-called restricted isometry property~\cite{candes2008restricted}, 
or their covariance matrix can be accurately estimated~\cite{negahban2011estimation}. Therefore, their sample complexity $n$ depends 
on either $d_1$ or the smallest eigenvalue value $\lambda_{\min}$ of $\mx$'s covariance matrix. For example, a mean-squared error (MSE) result that appeared in\ifnum\full=0
~\cite{negahban2011estimation} \else~\cite{negahban2011estimation}\fi  is of the form
$O\left(\frac{r(d_1+d_2)}{n \lambda^2_{\min}}\right)$.
When
$n \leq d_1/\lambda^2_{\min}$,
this result becomes trivial because the forecast $\hat \my = \mzero$ produces a comparable MSE. We design an efficient algorithm for the large feature size setting.  Our algorithm is a simple two-stage algorithm. Let $\mX \in \reals^{n \times d_1}$ be a matrix that stacks together all features and $\mY \in \reals^{n \times d_2}$ be the one that stacks the responses. In the first stage, we run a principal component analysis (PCA) on $\mX$ to obtain a set of uncorrelated features $\hat \mZ$. In the second stage, we run another PCA to obtain a low rank approximation of $\hat \mZ^{\transpose} \mY$ and use it to construct an output. 

While the algorithm is operationally simple, we show a powerful and generic result on using PCA to process features, a widely used practice for ``dimensionality reduction''~\cite{cao2003comparison,ghodsi2006dimensionality,friedman2001elements}. PCA is known to be effective to orthogonalize features by keeping 
only the subspace explaining large variations. But its performance can only be analyzed under the so-called factor model~\cite{stock2002forecasting,stock2011dynamic}. We show the efficacy of PCA \emph{without} the factor model assumption. 
Instead, PCA should be interpreted as a robust estimator of $\mx$'s covariance matrix. The empirical estimator $C = \frac 1 n \mX \mX^{\transpose}$ in the high-dimensional setting cannot be directly used because $n \ll d_1 \times d_2$, but it exhibits an interesting regularity: the leading eigenvectors of $C$ are closer to ground truth than the remaining ones. In addition,
the number of reliable eigenvectors grows as the sample size grows, so our PCA procedure
projects the features along reliable eigenvectors and 
\emph{dynamically} adjusts $\hat \mZ$'s rank to maximally utilize the raw features. 
Under mild conditions on the ground-truth covariance matrix $C^*$ of $\mx$, we show that it is always possible to decompose $\mx$ into a set of near-independent features and a set of (discarded) features that have an inconsequential impact on a model's MSE. 

When features $\mx$ are transformed into uncorrelated ones $\mz$, our original problem becomes $\my = N\mz + \epsilon$, which can be reduced to a matrix denoising problem~\cite{donoho2014minimax} and be solved by the second stage.  
Our algorithm guarantees that we can recover all singular vectors of $N$ whose associated singular values are larger than a certain threshold $\tau$. The performance guarantee can  be  translated into MSE bounds parametrized by commonly used variables (though, these translations usually lead to looser bounds). For example, when $N$'s rank is $r$, our result reduces 
the MSE from $O(\frac{r(d_1+d_2)}{n\lambda^{2}_{\min}})$  to 
 $O(\frac{rd_2}{n} + n^{-c})$ for a suitably small constant $c$. The improvement is most pronounced when $n \ll d_1$. 
 
We also provide a new matching lower bound. Our lower bound asserts that \emph{no algorithm} can recover a fraction of singular vectors of $N$ whose associated singular values are smaller than $\rho \tau$, where $\rho$ is a ``gap parameter''. Our lower bound contribution is twofold. First, we introduce a notion of ``local minimax'', which enables us to define a lower bound parametrized by the singular values of $N$. This is a stronger lower bound than those delivered by the standard minimax framework, which are often parametrized by the rank $r$ of $N$~\cite{koltchinskii2011nuclear}. Second, we develop a new probabilistic technique for establishing lower bounds under the new local minimax framework. Roughly speaking, 
our techniques assemble a large collection of matrices that share the same singular values of $N$ but are far from each other, so no algorithm can successfully distinguish these matrices with identical spectra.



\section{Preliminaries}

\vspace{-1mm}
\myparab{Notation.} Let $\mX \in \reals^{n \times d_1}$ and $\mY \in \reals^{n \times d_2}$ be data matrices with their $i$-th rows representing the $i$-th observation.
For matrix $A$, we denote
its singular value decomposition as $A = U^A \Sigma^A (V^A)^\transpose$ and $\mP_r(A)
\triangleq U_r^A \Sigma^A_r {V_r^A}^\transpose$ is the rank $r$ approximation obtained by
keeping the top $r$ singular values and the corresponding singular vectors. When the context is clear, we drop the superscript $A$ and use $U, \Sigma,$ and $V$ ($U_r$, $\Sigma_r$, and $V_r$) instead. Both $\sigma_i(A)$ and $\sigma^A_i$ are used to refer to $i$-th singular value of $A$. 
We use MATLAB notation when we refer to a specific row or column, e.g., $V_{1, :}$ is the first row of $V$ and $V_{:, 1}$ is the first column. $\|A\|_F$, $\|A\|_2$, and $\|A\|_{*}$ are Frobenius, spectral, and nuclear norms of $A$. In general, we use boldface upper case (e.g., $\mX$) to denote data matrices and boldface lower case (e.g., $\mx$) to denote one sample. Regular fonts denote other matrices. Let $C^* = \E[\mx \mx^{\transpose}]$ and $C = \frac 1 n \mX^{\transpose}\mX$ be the empirical estimate of $C^*$. Let $C^* = V^* \Lambda^* (V^*)^{\transpose}$ be the eigen-decomposition of the matrix $C^*$, and $\lambda^*_1
\geq \lambda^*_2, \ldots, \geq  \lambda^*_{d_1} \geq 0$ be the diagonal entries
of $\Lambda^*$. Let $\{\mmu_1, \mmu_2, \dots \mmu_{\ell}\}$ be an arbitrary set of column vectors, and $\spans(\{\mmu_1, \mmu_2, \dots, \mmu_{\ell}\})$ be the subspace spanned by it. 
An event happens with high probability means that it happens with probability $\geq 1 - n^{-5}$, where 5 is an arbitrarily chosen large constant and is not optimized. 
\begin{figure*}[th!]
\begin{subfigure}[t]{0.5\textwidth}
{
\begin{codebox}
\Procname{$\proc{Step-1-PCA-X}(\mX)$}
\li $[U, \Sigma, V] = \mathrm{svd}(\mX)$
\li $\Lambda = \frac 1 n (\Sigma^2)$; $\lambda_i = \Lambda_{i,i}$. 
\li \Comment \textbf{\small Gap thresholding.}
\li \Comment {\small $\delta = n^{-O(1)}$ is a tunable parameter.}
\li $k_1 = \max\{k_1: \lambda_{k_1} - \lambda_{k_1 + 1} \geq \delta\}$,
\li $\Lambda_{k_1}$: {\small diagonal matrix comprised of $\{\lambda_i\}_{i \leq {k_1}}$}. 
\li $U_{k_1}, V_{k_1}$: {\small $k_1$ leading columns of $U$ and $V$.}
\li $\hat \Pi = (\Lambda_{k_1})^{-\frac 1 2}V^{\transpose}_{k_1}$
\li $\hat \mZ_{+} = \sqrt n U_{k_1} (= X \hat \Pi^{\transpose})$. 
\li \Return $\{\hat \mZ_+, \hat \Pi\}$. 
\end{codebox}

}
\end{subfigure}
\begin{subfigure}[t]{0.5\textwidth}
{
\begin{codebox}
\Procname{$\proc{Step-2-PCA-Denoise}(\hat \mZ_+, \mY)$}
\li $\hat N^{\transpose}_+ \gets \frac 1 n\hat \mZ_+^{\transpose}\mY$. 
\li \Comment \textbf{\small Absolute value thresholding.}
\li \Comment {\small $\theta$ is a suitable constant; $\sigma_{\epsilon}$ is std. of the noise.}
\li $k_2 = \max \Bigl\{k_2: \sigma_{k_2}(\hat N_+) \geq \theta \sigma_{\epsilon}\sqrt{\frac{d_2}{n}}\Bigr\}$.
\li \Return $\mP_{k_2}(\hat N_+)$ 
\end{codebox}
\begin{codebox}
\Procname{$\proc{Adaptive-RRR}(\mX, \mY)$}
\li $[\hat Z_+, \hat \Pi] = \proc{Step-1-PCA-A}(\mX)$.
\li $\mP_{k_2}(\hat N_+) = \proc{Step-2-PCA-Denoise}(\hat Z_+, \mY)$. 
\li \Return $\hat M = \mP_{k_2}(\hat N_+) \hat \Pi$
\end{codebox}
}
\end{subfigure}
\vspace{-2mm}
\caption{Our algorithm ($\proc{Adaptive-RRR}$) for solving the regression $\my = M \mx + \epsilon$. 
}

\label{fig:pcarrr}
\vspace{-.6cm}
\end{figure*}

\myparab{Our model.} We consider the model $\my = M \mx + \epsilon$, where $\mx \in \reals^{d_1}$ is a multivariate Gaussian, $\my \in \reals^{d_2}$, $M\in\reals^{d_2 \times d_1}$, and $\epsilon \sim N(0, \sigma^2_{\epsilon}I_{d_2 \times d_2})$. We can relax the Gaussian assumptions on $\mx$ and $\epsilon$ for most results we develop. We assume a PAC learning framework, i.e., we observe a sequence $\{(\mx_i, \my_i)\}_{i \leq n}$ of independent samples and our goal is to find an $\hat M$ that minimizes the test error $\E_{\mx, \my}[\|\hat M \mx - M\mx\|_2^2]$. 
We are specifically interested in the setting in which $d_2 \approx n \leq d_1$.

The key assumption we make to circumvent the $d_1 \geq n$ issue is that the features are correlated. This assumption can be justified for the following reasons: \emph{(i)} In practice, it is difficult, if not impossible, 
to construct completely uncorrelated features. \emph{(ii)} When $n \ll d_1$, it is not even possible to \emph{test} whether the features are uncorrelated~\cite{bai2010spectral}. 
\emph{(iii)} When we indeed know that the features are independent, there are significantly simpler methods to design models. For example, we can build multiple models such that each model regresses on an individual feature of $\mx$, and then use a boosting/bagging method~\cite{friedman2001elements,schapire2013boosting} to consolidate the predictions. 

The correlatedness assumption implies that the eigenvalues of $C^*$ \emph{decays}. The only (full rank) positive semidefinite matrices that have non-decaying (uniform) eigenvalues are the identity matrix (up to some scaling). In other words, when $C^*$ has uniform eigenvalues, $\mx$ has to be uncorrelated.

We aim to design an algorithm that works \emph{even when} the decay is slow, such as when $\lambda_i(C^*)$ has a heavy tail. Specifically, our algorithm assumes $\lambda_i$'s are bounded by a heavy-tail power law series:

\begin{assumption} The $\lambda_i(C^*)$ series satisfies
$\lambda_i(C^*) \leq c \cdot i^{-\omega}$ for a constant $c$ and $\omega \geq 2$.  \end{assumption}
\vspace{-2mm}
We \emph{do not} make functional form assumptions on $\lambda_i$'s. This assumption also covers many benign cases, such as when $C^*$ has low rank or its eigenvalues decay exponentially. Many empirical studies report power law distributions of data covariance matrices~\cite{akemann2010universal,nobi2013random,urama2017analysis,clauset2009power}. Next, we make standard normalization assumptions. $\E\|\mx\|^2_2 = 1$, $\|M\|_2 \leq \Upsilon = O(1)$, and $\sigma_{\epsilon} \geq 1$.  
Remark that we assume only the spectral norm of $M$ is bounded, while its Frobenius norm can be unbounded. Also, we assume the noise $\sigma_{\epsilon} \geq 1$ is sufficiently large, which is more important in practice. The case when $\sigma_{\epsilon}$ is small can be tackled in a similar fashion. 
Finally, our studies avoid examining excessively unrealistic cases, so we assume $d_1 \leq d^3_2$. We examine 
the setting where existing algorithms fail to deliver non-trivial MSE, so we assume that $n \leq r d_1 \leq d^4_2$.

\section{Upper bound}

Our algorithm (see Fig.~\ref{fig:pcarrr}) consists of two steps.  
\ifnum\full=0
\else
\fi
\myparab{Step 1. Producing uncorrelated features.} We run a PCA to obtain a total number of $k_1$ orthogonalized features. See $\proc{Step-1-PCA-X}$ in Fig.~\ref{fig:pcarrr}. Let the SVD of $\mX$ be $\mX = U\Sigma(V)^{\transpose}$. Let $k_1$ be a suitable rank chosen by inspecting the gaps of $\mX$'s singular values (Line 5 in $\proc{Step-1-PCA-X}$). $\hat \mZ_+ =\sqrt n U_{k_1}$ is the set of transformed features output by this step. The subscript $+$ in $\hat \mZ_+$ reflects that a dimension reduction happens so the number of columns in $\hat \mZ_+$ is smaller than that in $\mX$. Compared to standard PCA dimension reduction, there are two differences: \emph{(i)} We use the left leading singular vectors of $\mX$ (with a re-scaling factor $\sqrt n$) as the output, whereas the PCA reduction outputs $\mP_{k_1}(\mX)$. \emph{(ii)} We design a specialized rule to choose $k_1$ whereas PCA usually uses a hard thresholding or other ad-hoc rules. 
\myparab{Step 2. Matrix denoising.} We run a second PCA on the matrix $(\hat N_+)^{\transpose} \triangleq \frac 1 n \hat{\mZ}_{+}^{\transpose} \mY$. The rank $k_2$ is chosen by a hard thresholding rule (Line 4 in $\proc{Step-2-PCA-Denoise}$). Our final estimator is $\mP_{k_2}(\hat N_+)\hat \Pi$, where $\hat \Pi = (\Lambda_{k_1})^{-\frac 1 2}V^{\transpose}_{k_1}$ is computed in $\proc{Step-1-PCA-X}(\mX)$.

\subsection{Intuition of the design}\label{sec:intuition}

While the algorithm is operationally simple, its design is motivated by carefully unfolding the statistical structure of the problem. We shall realize that applying PCA on the features  \emph{should not} be viewed as removing noise from a factor model, or finding subspaces that maximize variations explained by the subspaces as suggested in the standard literature~\cite{friedman2001elements,stock2002forecasting,stock2005implications}. 
Instead, it implicitly implements a robust estimator for $\mx$'s precision matrix, and the design of the estimator needs to be coupled with our objective of forecasting $\my$, thus resulting in a new way of choosing the rank.

\myparab{Design motivation: warm up.} We first examine a simplified problem $\my = N \mz + \epsilon$, where variables in $\mz$ are assumed to be uncorrelated. Assume $d = d_1 = d_2$
in this simplified setting. Observe that 
\begin{align}\label{eqn:denoise}
\frac 1 n \mZ^{\transpose} \mY  = \frac 1 n \mZ^{\transpose}(\mZ N^{\transpose} + E) =  (\frac 1 n \mZ^{\transpose} \mZ )N^{\transpose}+ \frac 1 n\mZ^{\transpose} E \approx
I_{d_1 \times d_1} N^{\transpose} + \frac{1}{n} \mZ^{\transpose} E =
N^{\transpose} +  \calE, 
\end{align}
where $E$ is the noise term and $\calE$ can be approximated by a matrix with independent zero-mean noises. 

\vspace{-2mm}
\mypara{Solving the matrix denoising problem.} Eq.~\ref{eqn:denoise} implies that when we compute $\mZ^{\transpose}\mY$, the problem reduces to an extensively studied matrix denoising problem~\cite{donoho2014minimax,gavish2014optimal}. 
We include the intuition for solving this problem for completeness. The signal $N^{\transpose}$ is overlaid
with a noise matrix $\calE$. $\calE$ will elevate all the singular values of $N^{\transpose}$ by an order of $\sigma_{\epsilon}\sqrt{d/n}$. 
We run a PCA to extract reliable signals: 
when the singular value of a subspace is $\gg \sigma_{\epsilon}\sqrt{d/n}$, the subspace contains significantly more signal than noise and thus we keep the subspace. Similarly, 
a subspace associated a singular value $\lesssim \sigma_{\epsilon} \sqrt{d/n}$ mostly contains noise. This leads to a hard thresholding algorithm that sets 
 $\hat N^{\transpose} = \mP_r(N^{\transpose} + \calE)$, where $r$ is the maximum index such that $\sigma_r(N^{\transpose} + \calE) \geq c\sqrt{d/n}$ for some constant $c$. 
In the general setting $\my = M \mx + \epsilon$, $\mx$ may not be uncorrelated. But when we set $\mz = (\Lambda^*)^{-\frac 1 2}(V^*)^{\transpose}\mx$, we see that $\E[\mz \mz^{\transpose}] = I$. This means knowing $C^*$ suffices to reduce the original problem to a simplified one. Therefore, our algorithm uses Step 1 to estimate $C^*$ and $\mZ$, and uses Step 2 to reduce the problem to a matrix denoising one and solve it by standard thresholding techniques.

\myparab{Relationship between PCA and precision matrix estimation.} In step 1, while we plan to estimate $C^*$, our algorithm runs a PCA on $\mX$. 
We observe that empirical covariance matrix $C = \frac 1 n \mX^{\transpose} \mX = \frac 1 n V (\Sigma)^2 (V)^{\transpose}$, 
i.e., $C$'s eigenvectors coincide with $\mX$'s right singular vectors. When we use the empirical estimator to construct $\hat \mz$, we obtain $\hat \mz = \sqrt n (\Sigma)^{-1}(V)^{\transpose}\mx$. When we apply this map to every training point and assemble the new feature matrix, we exactly get 
$\hat \mZ = \sqrt n \mX V (\Sigma)^{-1} = \sqrt{n} U$.  It means that using $C$ to construct $\hat \mz$ is the same as running a PCA in  $\proc{Step-1-PCA-X}$
with $k_1 = d_1$. 

\begin{wrapfigure}{l}{6.5cm}
\vspace{-6mm}
\includegraphics[width=\linewidth]{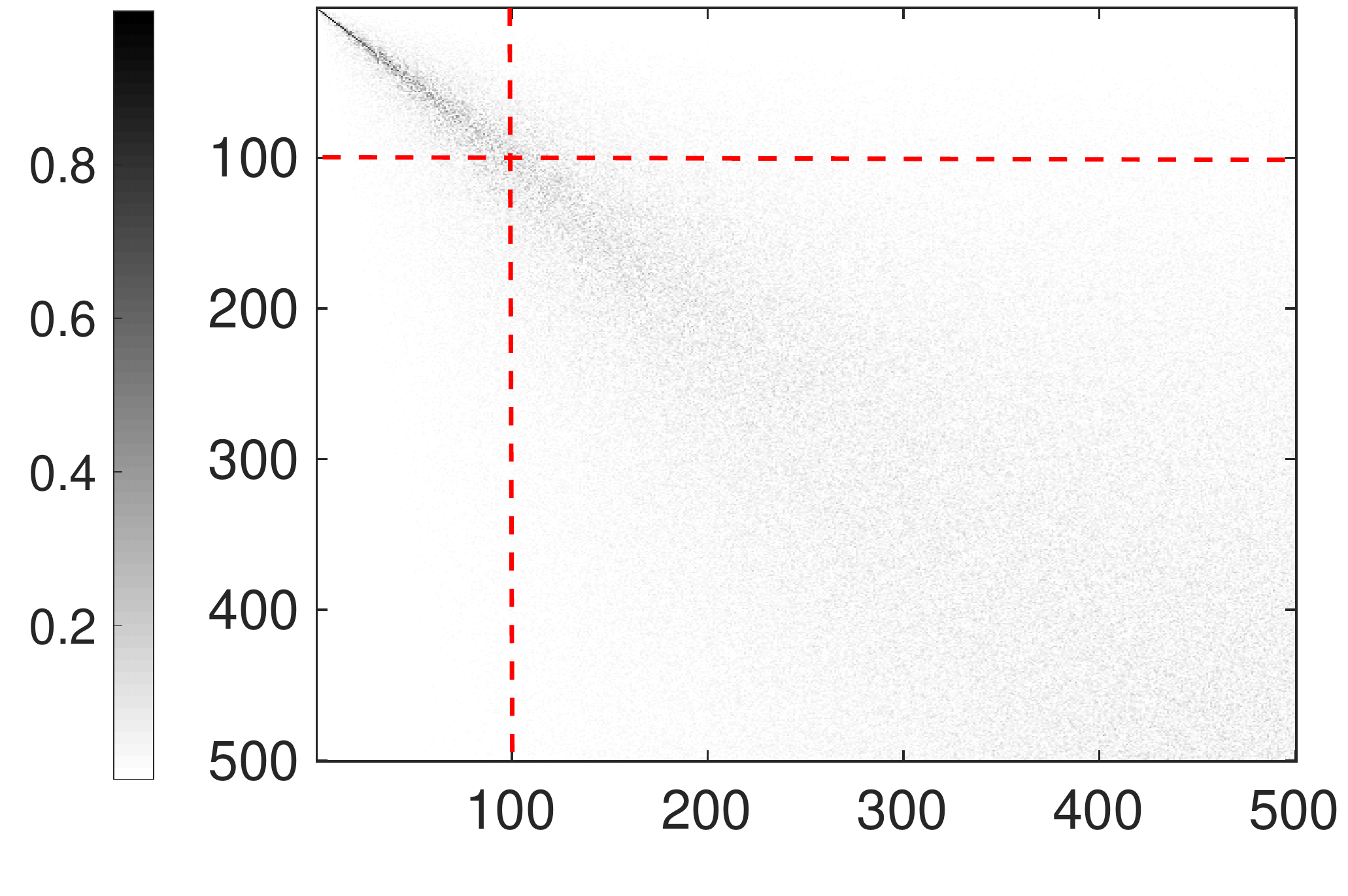}
\captionsetup{font=small}
\caption{The angle matrix between $C$ and $C^*$.}\vspace{-.3cm}
\label{fig:est_quality}
\vspace{-2mm}
\end{wrapfigure}

When $k_1 < d_1$, PCA uses a low rank approximation of $C$ as an estimator for $C^*$. We now explain why this is effective. First, note that $C$ is \emph{very far} from $C^*$ when $n \ll d_1$, therefore it is dangerous to directly plug in $C$ to find $\hat \mz$. Second, an interesting regularity of $C$ exists and can be best explained by a picture. In Fig.~\ref{fig:est_quality}, we plot the pairwise angles 
between eigenvectors of $C$ and those of $C^*$ from a synthetic dataset. Columns are sorted by the $C^*$'s eigenvalues in decreasing order. When $C^*$ and $C$ coincide, this plot would look like an identity matrix. When $C$ and $C^*$ are unrelated, then the plot behaves like a block of white Gaussian noise. We observe a pronounced pattern: the angle matrix can be roughly divided into two sub-blocks (see the red lines in Fig.~\ref{fig:est_quality}). The upper left sub-block behaves like an identity matrix, suggesting that the leading eigenvectors of $C$ are close to those of $C^*$. The lower right block behaves like a white noise matrix, suggesting that the ``small'' eigenvectors of $C$ are far from those of $C^*$. When $n$ grows, one can observe the upper left block becomes larger and this the eigenvectors of $C$ will sequentially get stabilized. Leading eigenvectors are first stabilized, followed by smaller ones. 
Our algorithm leverages this regularity by keeping only a suitable number of reliable eigenvectors from $C$ while ensuring not much information is lost when we throw away those ``small'' eigenvectors.

\myparab{Implementing the rank selection.} We rely on three interacting building blocks:
\vspace{-2mm}

\mypara{1. Dimension-free matrix concentration.} First, we need to find a concentration behavior of $C$ for $n \leq d_1$ to decouple $d_1$ from the MSE bound. We utilize a 
 dimension-free matrix concentration inequality~\cite{oliveira2010sums}\ifnum\full=1
 (also replicated as Lemma~\ref{lem:chernoff})
 \fi.
Roughly speaking, the concentration behaves as $\|C - C^*\|_2 \approx n^{-\frac 1 2}$.  
This guarantees that $|\lambda_i(C) - \lambda_i(C^*)| \leq n^{-\frac 1 2}$ by standard matrix perturbation results~\cite{kato1987variation}.

\vspace{-2mm}
\mypara{2. Davis-Kahan perturbation result.} However, the pairwise closeness of the $\lambda_i$'s does not imply the eigenvectors are also close.  
When $\lambda_i(C^*)$ and $\lambda_{i+1}(C^*)$ are close, the corresponding eigenvectors in $C$ can be ``jammed'' together. Thus, we need to identify 
an index $i$, at which $\lambda_i(C^*) - \lambda_{i + 1}(C^*)$ exhibits significant gap, and use a Davis-Kahan result to show that $\mP_i(C)$ is close to $\mP_{i}(C^*)$. 
On the other hand, the map $\Pi^* (\triangleq (\Lambda^*)^{-\frac 1 2}(V^*)^{\transpose})$ we aim to find depends on the square root of inverse $(\Lambda^*)^{-\frac 1 2}$, so we need 
additional manipulation to argue our estimate is close to $(\Lambda^*)^{-\frac 1 2}(V^*)^{\transpose}$. 

\vspace{-2mm}
\mypara{3. The connection between gap and tail.} Finally,  the performance of our procedure 
is also characterized by the total volume of signals that are discarded, i.e., $\sum_{i > k_1}\lambda_i(C^*)$, where $k_1$ is the location that exhibits the gap. 
The question becomes whether it is possible to identify a $k_1$ that simultaneously exhibits a large gap and ensures the tail after it is well-controlled, e.g., the sum of the tail is $O(n^{-c})$ for a constant $c$. 
We develop a combinatorial analysis to show that it is \emph{always possible} to find such a gap under the assumption that $\lambda_i(C^*)$ is bounded by a power law distribution with exponent $\omega \geq 2$. 
Combining all these three building blocks, we have:

\begin{proposition}\label{prop:closeness} Let $\xi$ and $\delta$ be two tunable parameters such that $\xi = \omega(\log^3 n / \sqrt n)$ and $\delta^3 = \omega(\xi)$. Assume that $\lambda^*_i \leq c \cdot i^{-\omega}$. Consider running $\proc{Step-1-PCA-X}$ in Fig.~\ref{fig:pcarrr}, with high probability, we have
(i) \text{Leading eigenvectors/values are close:} there exists a unitary matrix $W$ and a constant $c_1$ such that 
$ \|V_{k_1}(\Lambda_{k_1})^{-\frac 1 2} - V^*_{k_1}(\Lambda^*_{k_1})^{- \frac 1 2}W\| \leq \frac{c_1 \xi}{\delta^3}. $
(ii) \text{Small tail:}
 $   \sum_{i \geq k_1} \lambda^*_i \leq c_2 \delta^{\frac{\omega - 1}{\omega + 1}}$ for a constant $c_2$. 
\end{proposition}
\vspace{-3mm}
Prop.~\ref{prop:closeness} implies that our estimate $\hat \mz_+ = \hat \Pi(\mx)$ is sufficiently close
to $\mz = \Pi^*(\mx)$, up to a unitary transform. We then execute $\proc{Step-2-PCA-Denoise}$ to reduce the problem to a matrix denoising one and solve it by hard-thresholding. 
Let us refer to $\my = N \mz + \epsilon$, where $\mz$ is a standard multivariate Gaussian and $N = MV^*(\Lambda^*)^{\frac 1 2}$ as the \emph{orthogonalized form} of the problem. 
While we do not directly observe $\mz$, our performance is characterized by spectra structure of $N$. 

\begin{theorem}\label{thm:main_upper} Consider running $\proc{Adaptive-RRR}$ in Fig.~\ref{fig:pcarrr} on $n$ independent samples $(\mx, \my)$ from the model $\my = M \mx + \epsilon$, where $\mx \in \reals^{d_1}$ and $\my \in \reals^{d_2}$. Let 
$C^* = \E[\mx \mx^{\transpose}]$.
Assume that 
\emph{(i)} $\|M\|_2 \leq \Upsilon = O(1)$, and \emph{(ii)} 
 $\mx$ is a multivariate Gaussian with $\|\mx\|_2 = 1$. In addition, $\lambda_1(C^*) < 1$ and for all $i$, $\lambda_i(C^*) \leq c/i^{\omega}$ for a constant $c$, and \emph{(iii)} $\epsilon \sim N(0, \sigma^2_{\epsilon} I_{d_1})$, where $\sigma_{\epsilon} \geq \min\{\Upsilon, 1\}$. 
 
 Let $\xi = \omega(\log^3n / \sqrt{n})$, $\delta^3 = \omega(\xi)$, and $\theta$ be a suitably large constant.
 Let $\my = N \mz + \epsilon$ be the orthogonalized form of the problem.
 Let $\ell^*$ be the largest index such that $\sigma^N_{\ell^*} > \theta \sigma_{\epsilon}\sqrt{\frac{d_2}{n}}$. Let $\hat \my$ be our testing forecast.  With high probability over the training data:
 \ifnum\full=0

\begin{equation}
  \E[\|\hat \my - \my\|^2_2]    \leq \sum_{i > \ell^*}(\sigma^N_i)^2 + O\left(\frac{\ell^* d_2 \theta^2 \sigma^2_{\epsilon}}{n}\right)  + O\left(\sqrt{\frac{\xi}{\delta^3}}\right) + O\left(\delta^{\frac{\omega - 1}{4(\omega+1)}}\right)\label{eqn:ms}
\end{equation}

\else
\begin{equation}
  \E[\|\hat \my - \my\|^2_2]   \leq \sum_{i \geq \ell^*}(\sigma^N_i)^2 + O\left(\frac{\ell^* d_2 \theta^2 \sigma^2_{\epsilon}}{n}\right) + O\left(\sqrt{\frac{\xi}{\delta^3}}\right) + O\left(\delta^{\frac{\omega - 1}{4(\omega+1)}}\right)\label{eqn:ms}
\end{equation}
\fi

The expectation is over the randomness of the test data.
\end{theorem}

Theorem~\ref{thm:main_upper} also implies that there exists a way to parametrize $\xi$ and $\delta$ such that 
$ \E[\|\hat \my - \my\|^2_2] \leq \sum_{i > \ell^*}(\sigma^N_i)^2 + O\left(\frac{\ell^* d_2 \theta^2 \sigma^2_{\epsilon}}{n}\right) + O(n^{-c_0}) 
$
for some constant $c_0$. We next interpret each term in  (\ref{eqn:ms}). 

\mypara{Terms $\sum_{i > \ell^*}(\sigma^N_i)^2+O\left(\frac{\ell^* d_2 \theta^2 \sigma^2_{\epsilon}}{n}\right)$} are typical for solving a matrix denoising problem $\hat N^{\transpose}_+ + \calE (\approx N^{\transpose} + \calE)$: we can extract signals associated with $\ell^*$ leading singular vectors of $N$, so  $\sum_{i > \ell^*}(\sigma^N_i)^2$ starts at $i > \ell^*$. For each direction we extract, we need to pay a noise term of order 
$\theta^2\sigma^2_{\epsilon}\frac{d_2}{n}$, leading to the term $O\left(\frac{\ell^* d_2 \theta^2 \sigma^2_{\epsilon}}{n}\right)$.
\mypara{Terms $O\left(\sqrt{\frac{\xi}{\delta^3}}\right) + O\left(\delta^{\frac{\omega - 1}{4(\omega+1)}}\right)$}
come from the estimations error of $\hat \mz_+$ produced from Prop.~\ref{prop:closeness}, consisting of both estimation errors of $C^*$'s leading eigenvectors and the error of cutting out a tail. 
We pay an exponent of $\frac 1 4$ on both terms (e.g., $\delta^{\frac{\omega - 1}{\omega+1}}$ in Prop.~\ref{prop:closeness} becomes $\delta^{\frac{\omega - 1}{4(\omega+1)}}$) because we used Cauchy-Schwarz (CS) twice. One is used in 
running matrix denoising algorithm with inaccurate $\mz_+$; the other one is used to bound the impact of cutting a tail. It remains open whether two CS is can be circumvented. 

Sec.~\ref{sec:lowerbound} explains how Thm~\ref{thm:main_upper} and the lower bound imply the algorithm is near-optimal. Sec.~\ref{sec:related} compares our result with existing ones under other parametrizations, e.g. $\Rank(M)$. 


\ifnum\full=1

\subsection{Analysis}
We now analyze our algorithm. Our analysis consists of three steps. In step 1, we prove Proposition~\ref{prop:closeness}. In step 2, we relate $\mZ^{\transpose}\mY$ with the product produced by our estimate $\hat \mZ^{\transpose}_+\mY$. In step 3, we prove Theorem~\ref{thm:main_upper}.

\subsubsection{Step 1. PCA for the features (proof of Proposition~\ref{prop:closeness})}
This section proves Proposition~\ref{prop:closeness}. 
We shall first show that the algorithm always terminates. We have the following lemma. 

\begin{lemma}\label{lem:apptailmain} Let $\{\lambda_i\}_{i \leq d}$ be a sequence such that $\sum_{i \leq n} \lambda_i = 1$, $\lambda_i \leq c i^{-\omega}$ for some constant $c$, $\omega \geq 2$, and $\lambda_1 < 1$. Define $\delta_i = \lambda_i - \lambda_{i + 1}$ for $i \geq 1$. Let $\ell_0$ be a sufficiently large number, and $c_1$ and $c_2$ are two suitable constants. Let $\ell$ be any number such that $\ell \geq \ell_0$. Let $\tau$ be any parameter such that $\tau < \rho - 1$. There exists an $i^*$ such that \emph{(i) Gap is sufficiently large:} $\delta_{i^*} \geq c_1 \cdot \ell^{-(\tau \omega/(\omega - 1) + 1)}$, and \emph{\emph{(ii) } tail sum is small:} $\sum_{i \geq i^*} \lambda_i \leq c_2 / \ell^{-\tau}$. 
\end{lemma}

We remark that Lemma~\ref{lem:apptailmain} will also be able to show part (ii) of Proposition~\ref{prop:closeness} (this will be explained below). This lemma also corresponds to the ``connection between gap and tail'' building block referred in Section~\ref{sec:intuition}. 


\begin{proof}[Proof of Lemma~\ref{lem:apptailmain}] Define the function $h(t) = \sum_{i \geq t}c/i^{\omega} = \frac{c_3 + o(1)}{t^{\omega - 1}}$ (by Euler–Maclaurin formula), where $o(1)$ is a function of $t$

Next, let us define 
\begin{align*}
i_1 & =  \min \left\{i^*: \sum_{i \leq i^*} \lambda_i\geq 1 - h({\ell}^{\frac{\tau}{\omega - 1}})\right\}  - 1\\
i_2 & =  \min \left\{i^*: \sum_{i \leq i^*} \lambda_i \geq 1 - \frac 1 2 \times h({\ell}^{\frac{\tau}{\omega - 1}}) \right\} - 1.
\end{align*}

Roughly speaking, we want to identify an $i_1$ such that $\sum_{i \leq i_1}\lambda_i$ is smaller than $1 - h\left(\ell^{\frac{\tau}{\omega - 1}}\right)$ but is as close to it as possible. We can interpret $i_2$ in a similar manner. $i_1, i_2 \geq 1$ because of the assumption $\lambda_1 < 1$. 

We can verify that $i_1 < \ell^{\frac{\tau}{\omega-1}}$ because $\sum_{i \leq \ell^{\frac{\tau}{\omega - 1}}}\lambda_{i} \geq 1 - h\left(\ell^{\frac{\tau}{\omega - 1}}\right)$. We can similarly verify that $i_2 < c_4 \ell^{\frac{\tau}{\omega - 1}}$ for some constant $c_4$. Now using 
\begin{align*}
    \sum_{i \leq i_2 + 1} \lambda_i \geq 1 - \frac{(c_3+o(1))\ell^{-\tau}}{2} \\
    \sum_{i \leq i_1}\lambda_i \leq 1 - (c_3+o(1))\ell^{-\tau}. 
\end{align*}
We may use an averaging argement and show that there exists an $i_3 \in [i_1 + 1, i_2 + 1]$ such that 
\begin{align*}
\lambda_{i_3} & \geq\frac{(c_3 + o(1))\ell^{-\tau}}{i_2 - i_1} \geq \frac{(c_3+o(1)) \ell^{-\tau}}{c_4 \ell^{\frac{\tau}{\omega - 1}}} \\
& \geq c_5 \ell^{-\tau - \frac{\tau}{\omega - 1}} = c_5 \ell^{-\frac{\tau \omega}{\omega - 1}}. 
\end{align*}
Note that $c_5 \ell^{-\frac{\tau \omega}{\omega - 1}} \geq 2 c\ell^{-\omega}$ because $\tau < \omega - 1$. Next, using that $\lambda_{\ell} \leq c/\ell^{\omega}$, we have 
\begin{equation}
    \lambda_{\ell} = \overbrace{\lambda_{i_3}}^{\geq c_5 \ell^{-\frac{\tau \omega}{\omega - 1}}} + (\lambda_{i_3 + 1} - \lambda_{i_3}) + \dots + (\lambda_{\ell} - \lambda_{\ell - 1}) \leq \overbrace{c/\ell^{\omega}}^{\leq \frac{c_5}{2} \cdot \ell^{-\frac{\tau \omega}{\omega - 1}}}. 
\end{equation}
This implies one of $(\lambda_{i_3} - \lambda_{i_3+1}), \dots , (\lambda_{\ell-1} - \lambda_{\ell})$ is at least $\frac{c_5} 2\cdot\ell^{-\frac{\tau \omega}{\omega - 1}}/\ell$. In other words, there exists an $i^* \in [i_3 + 1, \ell]$ such that 
$\lambda_{i^*} - \lambda_{i^*+1} \geq \frac{c_5}{2}\ell^{-\left(\frac{\tau \omega}{\omega - 1} + 1\right)}.$ Finally, we can check that 
\begin{equation}
    \sum_{i \geq i^*}\lambda_{i^*} \leq \sum_{i \geq i_1}\lambda_{i^*} \leq h\left(\ell^{\frac{\tau}{\omega - 1}}\right) \leq \frac{c_2}{\ell^{\tau}}. 
\end{equation}
\end{proof}
We apply Lemma~\ref{lem:apptailmain} by setting $\tau \rightarrow \omega - 1$. There is a parameter $\ell$ that we can tune, such that it is always possible to find an $i^*$ where $\delta_{i^*} \geq c_1 \ell^{-(\omega+1)}$ and $\sum_{i\geq i^*}\lambda_i \leq 1 - c_2 \ell^{-(\omega-1)}$. For any $\delta = o(1)$ (a function of $n$), we can set $\ell = \Theta\left(\left(\frac{1}{\delta}\right)^{\frac{1}{\omega + 1}}\right)$. In addition, $\sum_{i \geq k_1}\lambda_i = O\left(\delta^{\frac{\omega - 1}{\omega + 1}}\right)$. This also proves the second part of the Proposition.

It remains to prove part (i) of Proposition~\ref{prop:closeness}. It consists of three steps. 

\mypara{Step 1. Dimension-free Chernoff bound for matrices.} We first give a bound on $\|C^* - C\|_2$, which characterizes the tail probability by using the first and second moments of random vectors. This is the key device enabling us to meaningfully recover signals even when $n \ll d_1$. 

\begin{lemma}\label{lem:cconcentrate} Recall that $C^* = \E[\mx \mx^{\transpose}]$ and $C = \frac 1 n \mX^{\transpose} \mX$. For any $\xi > 0$, 
\begin{equation} 
    \Pr[\|C^* - C\|_2 \geq \xi]\leq (2n^2) \exp(-n\xi^2/(\log^4 n)) + n^{-10}. 
\end{equation}
\end{lemma}

The exponent 10 is chosen arbitrarily and is not optimized. 

\begin{proof}[Proof of Lemma~\ref{lem:cconcentrate}] We use the following specific form of Chernoff bound (\cite{oliveira2010sums})

\begin{lemma}\label{lem:chernoff} Let $\mz_1$, $\mz_2, \dots, \mz_n$ be i.i.d. random vectors such that $\|\mz_i\| \leq \alpha$ a.s. and $\|\E[\mz_i \mz^{\transpose}_i]\| \leq \beta$. Then for any $\epsilon > 0$, 
\begin{equation}
\Pr\left[\left\| \frac 1 n \sum_{i \leq n}\mz_i\mz^{\transpose}_i - \E[\mz_i \mz^{\transpose}_i]\right\|_2 \geq \xi\right]  \leq (2n^2) \exp\left(-\frac{n\xi^2}{16 \beta \alpha^2 + 8 \alpha^2 \xi}\right)
\end{equation}
\end{lemma}

We aim to use Lemma~\ref{lem:chernoff} to show Lemma~\ref{lem:cconcentrate} and we set $\mz_i = \mx_i$. But the $\ell_2$-norm of $\mz_i$'s are unbounded so we need to use a simple coupling technique to circumvent the problem. Specifically, let $c_0$ be a suitable constant and define 

\begin{equation}
    \tilde \mz_i = \left\{
    \begin{array}{ll}
        \mz_i & \mbox{ if $|\mz_i| \leq c_0 \log^2n$}  \\
        0 &  \mbox{ otherwise.}
    \end{array}
    \right.
\end{equation}

By using a standard Chernoff bound, we have 
\begin{equation}
    \Pr[\exists i: \tilde \mz_i \neq \mz_i] \leq \frac 1 {n^{10}}.
\end{equation}
Let us write $\tilde C = \frac 1 n \sum_{i \leq n}\tilde \mz_i \tilde \mz^{\transpose}_i$. We set $\alpha = c_0 \log^2n$ and $\beta = \Theta(1)$ in Lemma~\ref{lem:chernoff}. One can see that 
\begin{equation}
    \Pr[\|C^* - C\|_2 \geq \xi]\leq \Pr\left[(\|\tilde C - C\|_2 \geq \xi)\vee (\tilde C \neq C)\right] \leq 2n^2 \exp\left(-\frac{n \xi^2}{\log^4n}\right) + \frac{1}{n^{10}}.
\end{equation}
\end{proof}

\mypara{Step 2. Davis-Kahan bound.} The above analysis gives us that $\|C^* - C\|_2 \leq \xi$. We next show that the first a few eigenvectors of $C$ are close to those of $C^*$. 

\begin{lemma}\label{lem:ppclose} Let $\xi = \omega\left(\frac{\log^3 n}{\sqrt{n}}\right)$ and $\delta^3 = \omega(\xi)$. Considering running $\proc{Step-1-PCA-X}$ in Fig.~\ref{fig:pcarrr}. 
Let $\calp^* = V^*_{k_1}(V^*_{k_1})^{\transpose}$ and $\calp = V_{k_1}V^{\transpose}_{k_1}$.  When $\|C^* - C\|_2 \leq \xi$, 
$\|\calp^* - \calp\|_2 \leq \frac{2\xi}{\delta}.$
\end{lemma}
\begin{proof}
Recall that $\lambda^*_1, \lambda^*_2, \dots, \lambda^*_{d_1}$ are the eigenvalues of $C^*$. Let also $\lambda_1, \lambda_2, \dots , \lambda_{d_1}$ be the eigenvalues of $C$. Define

\begin{equation}
    S_1 = [\lambda_{k_1} - \delta/10, \infty] \quad \mbox{ and } \quad S_2 = [0, \lambda_{k_1 + 1} + \delta/10]. 
\end{equation}
The constant 10 is chosen in an arbitrary manner. Because $\|C^* - C\|_2 \leq \xi$, we know that $S_1$ contains $\lambda^*_{1}, \dots , \lambda^*_{k}$ and that $S_2$ contains $\lambda^*_{k_1 + 1}, \dots, \lambda^*_{d_1}$~\cite{kato1987variation}. Using the Davis-Kahan Theorem~\cite{davis1970rotation}, we get 
\begin{equation}
    \|\calp^* - \calp\|_2 \leq \frac{\|C^* - C\|_2}{0.8 \delta} \leq \frac{2\epsilon}{\delta}
\end{equation}
\end{proof}

We also need the following building block.

\begin{lemma} \label{lem:squareroot}~\cite{tang2013universally}
	Let $A$ and $B$ be $n \times n$ positive semidefinite matrices with the same
	rank of $d$. Let $X$ and $Y$ be of full column rank such that
	$XX^{\transpose} = A$ and $YY^{\transpose} = B$. Let $\delta$ be the
	smallest non-zero eigenvalue of $B$. Then there exists a unitary matrix $W \in \reals^{d \times d}$
	such that
	\begin{align*}
		\|XW - Y\|_2 \leq \frac{\|A - B\|_2(\sqrt{\|A\|_2} + \sqrt{\|B\|_2})}{\delta}. 
	\end{align*}
\end{lemma}

\mypara{Step 3. Commuting the unitary matrix.} Roughly speaking, Lemma~\ref{lem:ppclose} and Lemma~\ref{lem:squareroot} show that there exists a unitary matrix $W$ such that $\|V_{k_1}W - V^*_{k_1}\|_2$ is close to $0$. Standard matrix perturbation result also shows that $\Lambda_{k_1}$ and $\Lambda^*_{k_1}$ are close. This gives us that $V_{k_1}W \Lambda^{-\frac 1 2}_{k_1}$ and $V^*_{k_1}(\Lambda^*_{k_1})^{-\frac 1 2}$ are close, whereas we need that $V_{k_1}\Lambda^{-\frac 1 2}_{k_1}W$ and $V^*_{k_1}(\Lambda^*_{k_1})^{-\frac 1 2}$ are close. The unitary matrix $W$ is not in the right place. This is a standard technical obstacle for analyzing PCA based techniques~\cite{tang2013universally,fan2016overview,li2017world}. We develop the following lemma to address the issue.

\begin{lemma}\label{lem:aux-inv-app}
	Let $U_1, U_2$ be $n \times d$ matrices such that $U_1^\top U_1 = U_2^\top U_2 = I$. Let $S_1$, $S_2$ be diagonal matrices with strictly positive entries, and let $W \in \reals^{d \times d}$ be a unitary matrix. Then,
	\begin{align*}
		\| U_1 S_1^{-1} W - U_2 S_2^{-1} \| \leq \frac{\| U_1 S_1 W - U_2 S_2 \|}{\min\{(S_1)_{ii} \} \cdot \min\{(S_2)_{ii} \}} +  \frac{\|U_1 U_1^\top - U_2 U_2^\top\|}{\min\{(S_2)_{ii} \}}
	\end{align*}
\end{lemma}

\begin{proof}
	Observe that,
	\begin{align*} 
		U_1 S_1^{-1} W - U_2 S_2^{-1} &= U_1 S_1^{-1} W (S_2 U_2^\top - W^\top S_1 U_1^\top )U_2 S^{-1}_2  + U_1 U_1^\top U_2 S_2^{-1} - U_2 U_2^\top U_2 S_2^{-1}. 
	\end{align*}
	The result then follows by taking spectral norms of both sides, the triangle inequality and the sub-multiplicativity of the spectral norm.
\end{proof}

Results from Step 1 to Step 3 suffice to prove the first part of Proposition~\ref{prop:closeness}. First, we use Lemma~\ref{lem:ppclose} and Lemma~\ref{lem:squareroot} (adopted from ~\cite{tang2013universally}) to get that 
\begin{equation}
    \|V^*_{k_1} (\Lambda^*_{k_1})^{\frac 1 2} W - V_{k_1}(\Lambda_{k_1})^{\frac 1 2}\|_2 \leq \frac{c_0 \xi}{\delta^2}. 
\end{equation}
Next, observe that { $\lambda_{k_1}, \lambda^*_{k_1} = \Omega(\delta)$}. By applying Lemma~\ref{lem:aux-inv-app}, with $U_1 = V^*_{k_1}$ and $S_1 = (\Lambda^*_{k_1})^{\frac 1 2}$, $U_2 = V_{k_1}$ and $S_2 = (\Lambda_{k_1})^{\frac 1 2}$, we obtain 

\begin{align*}
	\| V_{k_1}^{ } \Lambda_{k_1}^{-\frac 1 2}W - V_{k_1}^* \Lambda_{k_1}^{*\,-\frac 1 2} \|_2 &\leq \frac{\| V_{k_1}^* \Lambda_{k_1}^{*\, \frac 1 2} W - V_{k_1}^{} \Lambda_{k_1}^{\frac{1}{2}} \|_2}{\delta} + \frac{\| \calP^* - \calP \|_2}{\delta} \leq \frac{c_1 \xi}{\delta^3} 
\end{align*}

This completes the proof of Proposition~\ref{prop:closeness}.

\subsubsection{Step 2. Analysis of $\mZ^{\transpose} \mY$}
\begin{proposition}
Consider running $\proc{Adaptive-RRR}$ in Fig.~\ref{fig:pcarrr} 
to solve the regression problem $\my = M\mx + \epsilon$. Let $\hat \mZ_+$ be the output of the first stage $\proc{Step-1-PCA-X}$. Let $W$ be the unitary matrix specified in Proposition~\ref{prop:closeness}. Let $\hat N_+ = \hat \mZ^{\transpose}_+ \mY$. We have with high probability (over the training data), 
\begin{align*}
    \hat N^{\transpose}_+ = W^{\transpose}N^{\transpose} + \calE_L + \calE_T, 
\end{align*}
where 
\begin{align*}
    \| \calE_L\|_2 \leq 2.2 \sigma_{\epsilon}\sqrt{\frac{d_2}{n}} \quad \mbox{ and } \quad \| \calE_T\|_F = O(\epsilon/\delta^3). 
\end{align*}
\label{unitary_noise}
\end{proposition}

The rest of this Section proves Proposition~\ref{unitary_noise}. Recall that $\mY = \mZ N^{\transpose} + E$ is the orthogonalized form of our problem. Let us split $N = [N_+, N_-]$, where $N_+ \in \reals^{d_2 \times k_1}$ consists of the $k_1$ leading columns of $N$ and 
$N_- \in \reals^{d_2 \times (d_1 - k_1)}$ consists of the remaining columns. Similarly, let $\mz = [\mz_+, \mz_-]$, where $\mz_+ \in \reals^{n \times k_1}$ and $\mz_- \in \reals^{n \times (d_1 - k_1)}$. Let $\mZ = [\mZ_+, \mZ_-]$, where $\mZ_+ \in \reals^{n \times k_1}$ and $\mZ_- \in \reals^{n \times (d_1 - k_1)}$. Finally, when we refer to estimated features of an individual instance produced from Step 1, we use $\hat \mz_+$. 

We have $\mY = \mZ_+ N^{\transpose}_+ + \mZ_- N^{\transpose}_- + E$. We let 
\begin{align*}
\delta_+ & = \hat \mz_+ - W^{\transpose} \mz_+ \\
\Delta_+ & = \hat \mZ_+  - \mZ_+ W, 
\end{align*}
where $\|\delta_+\|_2 = O(\epsilon/\delta^3)$ and $\|\Delta_+\|_2 = O(\sqrt n \epsilon/\delta^3)$.  We have 
\begin{align*}
\frac 1 n \hat \mZ^{\transpose}_+  \mY & = \frac 1 n (\Delta_+ + \mZ_+W)^{\transpose}(\mZ_+ N^{\transpose}_+ + \mZ_- N^{\transpose}_- + E)  \\
& =  W^{\transpose} N^{\transpose}_+ + W^{\transpose}\left(\frac 1 n \mZ^{\transpose}_+\mZ_{+} - I_{k_1 \times k_1}\right)N^{\transpose}_+ + \frac 1 n W^{\transpose}\mZ^{\transpose}_+\mZ_- N^{\transpose}_- + \frac 1 n W^{\transpose}\mZ^{\transpose}_+E \\
& \quad \quad +  \frac 1 n \Delta^{\transpose}_+ (\mZ_+ N^{\transpose}_+ + \mZ_- N^{\transpose}_- + E). 
\end{align*}

We shall let 
\begin{equation}
\hat N^{\transpose}_+ = W^{\transpose} N^{\transpose} + \calE,  
\end{equation}
where 
\begin{align*}
\calE & = \calE_1 + \calE_2 + \calE_3 + \calE_4 + \calE_5 \\
\calE_1 & = W^{\transpose}\left(\frac 1 n \mZ^{\transpose}_+ \mZ_+ - I_{k_1 \times k_1}\right)N^{\transpose}_+ \\
\calE_2 & = \frac 1 n W^{\transpose}\mZ^{\transpose}_+ \mZ_- N^{\transpose} \\
\calE_3 & = \frac 1 n W^{\transpose} \mZ^{\transpose}_+ E \\
\calE_4 & = \frac 1 n \Delta^{\transpose}_+ E \\
\calE_5 & = \frac 1 n \Delta^{\transpose}_+ (\mZ_+N^{\transpose}_+ + \mZ_- N^{\transpose}_-). 
\end{align*}

We next analyze each term. We aim to find bounds in either spectral norm or Frobenius norm. In some cases, it suffices to use $\|\calE_i\|_2 \cdot \Rank(\calE_i)$ to upper bound $\|\calE_i\|_F$. So we bound only $\calE_i$'s spectral norm. On the other hand, in the case of analyzing $\calE_5$, we can get a tight Frobenius norm bound but we cannot get a non-trivial spectral bound. 

From time to time, we will label the dimension of matrices in complex multiplication operations to enable readers to do sanity checks. 

\myparab{Bounding $\calE_1$.} We use the following Lemmas. 

\begin{lemma}\label{lem:zzp} Let $\mZ \in \reals^{n \times k_1}$, where $k_1 < n$. 
Let each entry of $\mZ$ be an independent standard Gaussian. 
We have 
\begin{equation}
    \left\|\frac 1 n \mZ^{\transpose} \mZ - I \right\| \leq \max \left\{\frac{10 \log^2n}{\sqrt n}, 4 \sqrt{\frac{k_1} n}\right\}
\end{equation}
\end{lemma}

\begin{proof}[Proof of Lemma~\ref{lem:zzp}] 
 We rely on the Lemma~\cite{rudelson2010non}: 

\begin{lemma} Let  $S \in \reals^{n \times k}$ ($n > k$) be a random matrix so that each $S_{i,j}$ is an independent standard Gaussian random variable. Let $\sigma_{\max}(S)$ be the maximum singular value of $S$ and $\sigma_{\min}(S)$ be the minimum singular value of it. We have 
\begin{equation}
    \Pr[\sqrt n - \sqrt k - t \leq \sigma_{\min}(S) \leq \sigma_{\max}(S) \leq \sqrt n + \sqrt k + t] \geq 1 - 2 \times \exp(-t^2/2). 
\end{equation}
\end{lemma}

We set $ t= \max\left\{\frac{\sqrt{k_1}}{10}, \log^2n\right\}$. 
Let us start with considering the case $\frac{\sqrt{k_1}}{10} > \log^2 n$. We have
\begin{equation}
    \sigma_{\min}(\mZ^{\transpose}\mZ) \geq n - 2.2 \sqrt{n k_1} + 1.21 k_1 \geq n - 2.2 \sqrt{n k_1}. 
\end{equation}
and 

\begin{equation}
    \sigma_{\max}(\mZ^{\transpose} \mZ)\leq n + 2.2 \sqrt{nk_1} + 1.21 k_1 \leq n + 4 \sqrt{n k_1}. 
\end{equation}

The case $\frac{\sqrt{k_1}}{10} \leq  \log^2 n$ can be analyzed in a similar fashion so that we can get 
\begin{equation}
    \left\|\frac 1 n \mZ^{\transpose}\mZ - I \right\| \leq \max \left\{\frac{10 \log^2 n}{\sqrt n}, 4 \sqrt{\frac{k_1}{n}}\right\}.
\end{equation}
\end{proof}

Therefore, we have 
\begin{align*}
\|\calE_1\|_2 \leq  \max \left\{\frac{10 \log^2 n}{\sqrt n}, 4 \sqrt{\frac{k_1}{n}}\right\} \|N^{\transpose}_+\|_2 = \Upsilon \max  \left\{\frac{10 \log^2 n}{\sqrt n}, 4 \sqrt{\frac{k_1}{n}}\right\}. 
\end{align*}

\myparab{Bounding $\calE_2$.} Observe that  $\E[\|\mZ_- N^{\transpose}_-\|^2_F] = n\|N^{\transpose}_-\|^2_F$. Also, 
\begin{align*}
\E[\|\underbrace{\mZ^{\transpose}_+}_{k_1 \times n} \underbrace{\mZ_-}_{ n \times (d_1 - k_1)} \underbrace{N^{\transpose}_-}_{(d_1-k_1)\times d_2}\|^2_F \mid \mZ_- N^{\transpose}_-] = k_1 \|\mZ_- N^{\transpose}_-\|^2_F. 
\end{align*}

Therefore, 
\begin{equation}\label{eq:cale2}
\E[\mZ^{\transpose}_+ \mZ_- N^{\transpose}_-] = k_1 n \|N^{\transpose}_-\|_F.
\end{equation}

We next bound $\|N_-\|_F$. 

\begin{lemma}\label{lem:Nminus} Let $N$ be the learnable parameter in normalized form $N = [N_+, N_-]$, where $N_+ \in \reals^{d_2 \times k_1}$ and $N_- \in \reals^{d_2 \times (d_1 - k_1)}$, and $k_1$ is determined by 
$\proc{Step-1-PCA-X}$. We have  $\|N_-\|_F = O\left(\delta^{\frac{\omega - 1}{\omega+1}}\right) = o(1)$.
\end{lemma}

\begin{proof}[Proof of Lemma~\ref{lem:Nminus}]
Recall that 
$$N = \underbrace{M}_{d_2 \times d_1} \underbrace{V^*}_{d_1 \times d_1} \underbrace{ (\Lambda^*)^{\frac 1 2}}_{d_1 \times d_1}.$$ 
We let $\Lambda^{*} = [\Lambda^*_+, \Lambda^*_-]$, where $\Lambda^*_+ \in \reals^{d_1 \times k_1}$ and $\Lambda^*_- \in \reals^{d_1 \times (d_1 - k_1)}$. We have $N_- = MV^* (\Lambda^*_-)^{\frac 1 2}$. Therefore, 
\begin{align}
\|N_-\|^2_F \leq \|M\|^2_2 \|V^*\|^2_2 \left\|(\Lambda^*)^{\frac 1 2}\right\|^2_F = O\left(\Upsilon \delta^{\frac{\omega - 1}{\omega+1}}\right) = o(1). \label{eqn:tailo}
\end{align}
Here, we used the assumption $\|M\|_2 = O(1)$ and the last equation holds because of Proposition~\ref{prop:closeness}.
\end{proof}

By (\ref{eq:cale2}), (\ref{eqn:tailo}), and a standard Chernoff bound, we have whp
\begin{align*}
\|\calE_2\|_2 \leq \|\calE_2\|_F  \leq 2 \sqrt{\frac{k_1}{n}}\|N_-\|_F= o\left(\sqrt{\frac{k_1}{n}}\right).
\end{align*}

\mypara{Bounding $\calE_3$. } We have the following Lemma. 

\begin{lemma}\label{lem:ze} Let $\mZ \in \reals^{n \times k_1}$ so that each entry in $\mZ$ is an independent standard Gaussian and 
$E \in \reals^{n \times d_2}$ so that each entry in $E$ is an independent Gaussian $N(0, \sigma^2_{\epsilon})$. For sufficiently large $n$, $k_1$, and $d_2$, where $k_1 \leq d_2$, 
we have 
$$\left\|\frac 1 n \mZ^{\transpose}E\right\| \leq \frac{1.1\sigma_{\epsilon}}{\sqrt n}(\sqrt{k_1} + \sqrt{d_2}).$$
\end{lemma}

\begin{proof}[Proof of Lemma~\ref{lem:ze}] Let $t =   \max \left\{\frac{10 \log^2 n}{\sqrt n}, 4 \sqrt{\frac{k_1}{n}}\right\}$. By Lemma~\ref{lem:zzp}, with high probability $\|\frac 1 n \mZ^{\transpose} \mZ - I\| \leq t$. This implies that the eigenvalues of $\mZ^{\transpose}\mZ$ are all within the range $n(1\pm t)$. Note that for $0 < \eta < 1/3$, if $\xi \in [1 - \eta, 1 + \eta]$, then $\sqrt{\xi} \in [1 - 2\eta, 1 + 2\eta]$. This implies that the singular values of $\vZ$ are within the range $\sqrt{n} (1 \pm 2 t)$.

Let $\Sigma^Z/\sqrt n = I + \Delta^Z$, where $\|\Delta^Z\|\leq 2t$. 
We have 
\begin{align}
	\frac 1 n \mZ^{\transpose} E & = 
	V^Z\left(\frac{\Sigma^Z}{\sqrt n}\right)(U^Z)^{\transpose}\frac{E}{\sqrt n} = V^Z(I+\Delta^Z)(U^Z)^{\transpose} \frac{E}{\sqrt n} \nonumber\\
	& = \underbrace{V^Z}_{k_1 \times k_1}\underbrace{(U^Z)^{\transpose}}_{k_1 \times n} \underbrace{\frac{E}{\sqrt n}}_{n \times d_2} + V^Z \Delta^Z (U^Z)^{\transpose} \frac{E}{\sqrt n}. \label{eqn:Zeps2}
\end{align}
Using the fact that the columns of $U^Z$ are orthonormal vectors, $V^Z$ is a unitary matrix, and $k_1 \leq d_2$, we see that $V^Z(U^Z)^{\transpose} E/\sqrt n$ is a matrix with i.i.d. Gaussian entries with standard deviation $\sigma_{\epsilon}/\sqrt n$. 

Let $B = V^Z(U^Z)^{\transpose} E/\sigma_{\epsilon}$ and $\tilde B =  (U^Z)^{\transpose} E/\sigma_{\epsilon}$. Then, from \eqref{eqn:Zeps2}, we have
\begin{equation}\label{eqn:couple}
	\frac 1 n \mZ^{\transpose} E  = \frac{\sigma_{\epsilon}}{\sqrt n}\left(B + V^Z \Delta^Z \tilde B\right). 
\end{equation}

The entries in $B$ ($\tilde B$) are all i.i.d Gaussian. By Marchenko-Pastar's law (and the finite sample bound of it~\cite{rudelson2010non}), we have with high probability
$\|\tilde B\|, \|B\|= \sqrt{k_1} + \sqrt{d_2} + o(\sqrt{k_1} + \sqrt{d_2})$.
Therefore, with high probability:
\begin{align*}
\left\|\frac 1 n \mZ^{\transpose}E \right\|_2 \leq \frac{1.1\sigma_{\epsilon}}{\sqrt n}(\sqrt{k_1} +\sqrt{d_2}).
\end{align*}

\end{proof}

Lemma~\ref{lem:ze} implies that 
\begin{align*}
\|\calE_3\|_2 \leq \frac{1.1 \sigma_{\epsilon}}{\sqrt n}(\sqrt{k_1} + \sqrt{d_2}). 
\end{align*}

\myparab{Bounding $\calE_4$.} We have 
\begin{align*}
\E[\|\calE_4 \|^2_F] & = \frac 1 {n^2} \E[\|\Delta^{\transpose}_+ E\|^2_F] = \frac{d_2}{n}\|\Delta_+\|^2_F = O\left(\frac{d_2\epsilon}{n \delta^3}\right) = o\left(\frac{d_2} n\right). 
\end{align*}

Using a Chernoff bound, we have whp $\|\calE_4\|_F = o\left(\frac{d_2}{n}\right)$. 

\myparab{Bounding $\calE_5$. } Because $\calE_5 = \frac 1 n \Delta^{\transpose}_+ (\mZ N^{\transpose})$, we have 
$$\|\calE_5 \|^2_F \leq \frac{1}{n^2} \|\Delta^{\transpose}_+\|^2_2 \|\mZ N^{\transpose}\|^2_F.$$
Using a simple Chernoff bound, we have whp, 
$$\|\mZ N^{\transpose}\|^2_F \leq 2 n \|M\mx\|_2^2 \leq 2n \|M\|^2_2 \|\mx\|^2_2 \leq 2 \Upsilon n.$$
This implies $\|\calE_5 \|_F^2 \leq O\left(\frac {\epsilon^2} {n^2\delta^6} n^2 \Upsilon^2 \right) = O\left(\frac{\epsilon^2}{\delta^6}\right).$

We may let 
\begin{align*}
\calE_L & =  \calE_1 + \calE_2 + \calE_3 + \calE_4  \\
\calE_T &  = \calE_5. 
\end{align*}

We can check that 
\begin{align*}
\|\calE_L\|_2  & \leq \|\calE_1\|_2 + \|\calE_2\|_2 + \|\calE_3\|_2 + \|\calE_4\|_2 \\
& \leq \Upsilon \max\left\{\frac{10 \log^2 n}{\sqrt n}, 4\frac{k_1}{n}\right\} + o\left(\sqrt{\frac{k_1}n}\right) + \frac{1.1 \sigma_{\epsilon}}{\sqrt n}(\sqrt{k_1} + \sqrt{d_2}) + o\left(\frac{d_2}{n}\right) \\
& \leq 2.2 \sigma_{\sigma_{\epsilon}} \sqrt{\frac{d_2}{n}}
\end{align*}

Also, we can see that $\|\calE_T\|_F = \|\calE_5 \|_F = O(\epsilon/\delta^3)$. This completes the proof for Proposition~\ref{unitary_noise}.

\subsubsection{Step 3. Analysis of our algorithm's MSE}
Let us recall our notation: 
\begin{enumerate}
    \item $\mz = (\Lambda^*)^{- \frac 1 2}(V^*)^{\transpose} \mx$ and $\delta_+  = \hat \mz_+ - W^{\transpose} \mz_+$.
    \item We let $\hat N^{\transpose}_+ = \hat \mZ^{\transpose}_+ \mY$ be the output of $\proc{Step-1-PCA-X}$ in Fig.~\ref{fig:pcarrr}. 
    \item All singular vectors in $\hat N_+$ whose associated singular values $\geq \theta \sigma_{\epsilon}\sqrt{\frac{d_2}{n}}$ are kept. 
\end{enumerate}

Let $\ell$ be the largest index such that $\sigma^{N_+}_{\ell} \geq \theta \sigma_{\epsilon}\sqrt{\frac{d_2}{n}}$. One can see that our testing forecast is  $\mP_{k_2}(\hat N_+) \hat \mz_+$. Therefore, we need to bound
 $\E_{\mz}[\|\mP_{k_2}(\hat N_+) \hat \mz_+ - N\mz\|^2]$. 
 
By Proposition~\ref{unitary_noise}, we have  $\hat N_+ = (W^{\transpose}N^{\transpose}_+ + \calE_L + \calE_T)^{\transpose},$ where  $\|\calE_L\|_2 \leq 2.2 \sigma_{\epsilon}\sqrt{\frac{d_2}{n}}$ and 
 $\|\calE_T\|_F = O(\xi/\delta^3)$ whp.  
 Let $\calE \triangleq \calE_L + \calE_T$.  We have
\begin{align*}
\mP_{k_2}(N_+ W+ \calE^{\transpose})\hat \mz_+ 
& = \mP_{k_2}(N_+ W + \calE^{\transpose}) (W^{\transpose}\mz_+ + \delta_+) \\
& = \mP_{k_2}(N_+ W + \calE^{\transpose})W^{\transpose}W (W^{\transpose}\mz_+ + \delta_+) \\
& = \mP_{k_2}((\underbrace{N_+}_{d_2 \times k_1} \underbrace{W}_{k_1 \times k_1} + \underbrace{\calE^{\transpose}}_{d_2 \times k_1})\underbrace{W^{\transpose}}_{k_1 \times k_1}) (\underbrace{W W^{\transpose}}_{k_1 \times k_1} \underbrace{\mz_+}_{k_1 \times 1} + \underbrace{W}_{k_1 \times k_1} \underbrace{\delta_+}_{k_1 \times 1}) \\
& = \mP_{k_2}\left(N_+ + (W\calE)^{\transpose} \right) (\mz_+ + W \delta_+). 
\end{align*}

Let $\calE' = (W\calE)^{\transpose}$, $\calE'_L = (W\calE_L)^{\transpose}$, $\calE'_T = (W \calE_{T})^{\transpose}$, and $\delta'_+ = W \delta_+$. We still have $\|\calE'_L\|_2 \leq 2.2 \sigma_{\epsilon}\sqrt{\frac{d_2}{n}}$, and $\|\calE'_T \|_F = O(\epsilon/\delta^3)$. 

We next have 
\begin{align*}
    & \E_{\mz}\left[\| \mP_{k_2}(\hat N_+) \hat \mz_+ - N\mz\|_2^2\right] & \\
    & = \E_{\mz}\Big[\big\|(\mP_{k_2}(N_+ + \calE') \mz_+ - N_+ \mz_+) + \mP_{k_2}(N_+ + \calE')\delta'_+ - N_-\mz_- \big\|^2_2\Big] & \\
    & \leq  \underbrace{\E_{\mz}\Big[\big\|(\mP_{k_2}(N_+ + \calE') \mz_+ - N_+ \mz_+)\big\|^2_2\Big]}_{\triangleq \Phi_1} + \underbrace{\E_{\mz}\Big[\big\|\mP_{k_2}(N_+ + \calE')\delta'_+ - N_-\mz_- \big\|^2_2\Big]}_{\triangleq \Phi_2} & \\
    & \quad \quad + 2 \sqrt{\E_{\mz}\Big[\big\|(\mP_{k_2}(N_+ + \calE') \mz_+ - N_+ \mz_+)\big\|^2_2\Big] \cdot \E_{\mz}\Big[\big\|\mP_{k_2}(N_+ + \calE')\delta'_+ - N_-\mz_- \big\|^2_2\Big]} & \\
    & \quad \quad \quad \quad \quad \quad \quad \quad \quad \quad \quad \quad \quad \quad \quad \quad  \mbox{(Cauchy Schwarz for random variables)} \\
    & = \Phi_1 + \Phi_2 + 2 \sqrt{\Phi_1 \Phi_2}. 
\end{align*}

We first bound $\Phi_2$ (the easier term). We have 
\begin{align*}
    \Phi_2 & = \E_{\mz}\Big[\big\|\mP_{k_2}(N_+ + \calE')\delta'_+ - N_-\mz_- \big\|^2_2\Big] \\
    & \leq 2 \E_{\mz}\Big[\big\|\mP_{k_2}(N_+ + \calE')\delta'_+\big\|^2_2\Big] + 2\E\Big[\big\|N_-\mz_- \big\|^2_2\Big]
\end{align*}

We first bound $\E_{\mz}\Big[\big\|\mP_{k_2}(N_+ + \calE')\delta'_+\big\|^2_2\Big]$. We consider two cases. 

\mypara{Case 1.} $\sigma_{\max}(N_+) > \frac{\theta}{2}\sigma_{\epsilon}\sqrt{\frac{d_2}{n}}$. In this case, 
we observe that $\|\calE\|_2 \leq 2.2 \sigma_{\epsilon}\sqrt{\frac{d_2}{n}} + o(1)$. This implies that $\|N_+ + \calE'\|_2 = O(\|N_+\|_2) = O(1)$. Therefore, $\E_{\mz}\Big[\big\|\mP_{k_2}(N_+ + \calE')\delta'_+\big\|^2_2\Big] \leq \|(N_+ + \calE')\delta'_+\|^2_2 = O(\|\delta'_+\|^2_2)$. 

\mypara{Case 2.}  $\sigma_{\max}(N_+) \leq \frac{\theta}{2}\sigma_{\epsilon}\sqrt{\frac{d_2}{n}}$. In this case, $\|N_+ + \calE'\|_2 \leq \theta \sigma\sqrt{\frac{d_2}{n}}$. This implies $\mP_{k_2}(N_++\calE') \delta'_+ = 0$ (i.e., the projection $\mP_{k_2}(\cdot)$ will not keep any subspace).  

This case also implies $\E_{\mz}\Big[\big\|\mP_{k_2}(N_+ + \calE')\delta'_+\big\|^2_2\Big] = 0 = O(\|\delta'_+\|^2_2)$

Next, we have $\E[\|N_- \mz_-\|^2_2] = \|N_-\|^2_F = O\left(\delta^{\frac{\omega - 1}{\omega + 1}}\right)$. 

Therefore, 
$$\Phi_2 = O\left(\frac{\epsilon^2}{\delta^6} + \delta^{\frac{\omega-1}{\omega+1}}\right).$$

Next, we move to bound 
$$\E_{\mz}\Big[\big\|(\mP_{k_2}(N_+ + \calE') \mz_+ - N_+ \mz_+)\big\|^2_2\Big].$$

We shall construct an orthonormal basis on $\reals^{d_2}$ and use the basis to ``measure the mass''. Let us describe this simple idea at high-level first. Let $\mmv_1, \mmv_2, \cdots, \mmv_{d_2}$ be a basis for $\reals^{d_2}$ and let $A \in \reals^{d_2 \times k_1}$ be an arbitrary matrix. We have $\|A\|^2_F = \sum_{i \leq d_2}\|\mmv^{\transpose}_iA\|^2_2$. The meaning of this equality is that we may apply a change of basis on the columns of $A$ and the ``total mass'' of $A$ should remain unchanged after the basis change. Our orthonormal basis consists of three groups of vectors.

\mypara{Group 1.} $\{U^{N_+}_{:, i}\}$ for $i \leq \ell$, where $\ell$ is the number of $\sigma_{i}(N^+)$ such that $\sigma_{i}(N^+) \geq \theta \sigma_{\epsilon}\sqrt{\frac{d_2}{n}}$. 

\mypara{Group 2.} The subspace in $\spans(\{U^{\hat N}_{:, i}\}_{i \leq k_2})$ that is orthogonal to $\{U^{N_+}_{:, i}\}_{i \leq \ell}$. Let us refer to these vectors as 
$\hat \mmu_1, \dots , \hat \mmu_s$ and $\hat U_{[s]} = [\hat \mmu_1, \dots \hat \mmu_s]$. 

\mypara{Group 3.} An arbitrary basis that is orthogonal to vectors in group 1 and group 2. Let us refer to them as $\mr_1, \dots, \mr_t$.

We have 
\begin{align*}
& \|\mP_{k_2}(N_++\calE') - N_+\|^2_F \\
& = \sum_{i \leq \ell}\left\|\left(U^{N_+}_{:, i}\right)^{\transpose}\left(\mP_{k_2}(N_+ + \calE') - N_+\right)\right\|^2_2 & \mbox{Term 1} \\
&  \quad  + \sum_{i \leq s} \| \hat \mmu_i^{\transpose}\left(\mP_{k_2}(N_+ + \calE') - N_+\right)\|^2_2  & \mbox{Term 2} \\
& \quad  + \sum_{i \leq r}\| \mr^{\transpose}_i \left(\mP_{k_2}(N_+ + \calE') - N_+\right)\|^2_2 & \mbox{Term 3}
\end{align*}

To understand the reason we perform such grouping, we can imagine making a decision for an (overly) simplified problem for each direction in the basis: consider a univariate estimation problem $y = \mu + \epsilon$ with $\mu$ being the signal, $\epsilon \sim N(0, \sigma^2)$ being the noise, and $y$ being the observation. Let us consider the case we observe only one sample. Now when $y \gg \sigma$, we can use $y$ as the estimator and $\E[(y - \mu)^2] = \sigma^2$. This high signal-to-noise setting corresponds to the vectors in Group 1. 

When $y \approx 3 \sigma$, we have $\mu^2 = \E[(y - \epsilon)^2] \approx (3-1)^2 \sigma^2 = 4\sigma^2$. On the other hand, $\E[(y - \mu)^2] = \sigma^2$. This means if we use $y$ as the estimator, the forecast is at least better than trivial. The median signal-to-noise setting corresponds to the vectors in group 2. 

When $y \ll \sigma$, we can simply use $\hat y = 0$ as the estimator. This low signal-to-noise setting corresponds to vectors in group 3. 

In other words, we expect: \emph{(i)} In term 1, signals along each direction of vectors in group 1 can be extracted. Each direction also pays a $\sigma^2$ term, which in our setting corresponds to $\theta \sigma_{\epsilon}\sqrt{\frac{d_2}{n}}$. Therefore, the MSE can be bounded by $O(\ell \theta^2\sigma^2_{\epsilon}d_2/n)$. \emph{(ii)} In terms 2 and 3, we do at least (almost) as well as the ``trivial forecast'' ($\hat \my = 0$). There is also an error produced by the estimator error from $\hat \mz_+$, and the tail error produced from cutting out features in $\proc{Step-1-PCA-X}$ in Fig.~\ref{fig:pcarrr}. 

Now we proceed to execute this idea. 

\myparab{Term 1. $\sum_{i \leq \ell}\left\|\left(U^{N_+}_{:, i}\right)^{\transpose}\left(\mP_{k_2}(N_+ + \calE') - N_+\right)\right\|^2_2$.}
 Let $\hat U \in \reals^{d_2 \times d_2}$ be the left singular vector of $N_+ + \calE'$. We let $\hat U$ have $d_2$ columns to include those vectors whose corresponding singular values are $0$ for the ease of calculation. We have 
\begin{align*}
& \sum_{i \leq \ell}\left\|\left(U^{N_+}_{:, i}\right)^{\transpose}\left(\mP_{k_2}(N_+ + \calE') - N_+\right)\right\|^2_2 \\
& = \sum_{i \leq \ell}\left\|\left(U^{N_+}_{:, i}\right)^{\transpose}\left(\hat U_{:, 1:k_2} \hat U^{\transpose}_{:, 1:k_2}(N_+ + \calE') - N_+\right)\right\|^2_2 \\
& = \sum_{i \leq \ell}\left\|\left(U^{N_+}_{:, i}\right)^{\transpose}\left(\left(\hat U \hat U^{\transpose } - \hat U_{:, k_2+1:d_2} \hat U^{\transpose}_{:, k_2+1:d_2}\right)(N_+ + \calE') - N_+\right)\right\|^2_2 \\
& = \sum_{i \leq \ell}\left\|\left(U^{N_+}_{:, i}\right)^{\transpose}\left(\calE' - \hat U_{:, k_2+1:d_2} \hat U^{\transpose}_{:, k_2+1:d_2} (N_+ + \calE')
\right)\right\|^2_2  \\
& \leq 2 \left\{\sum_{i \leq t}\left\|\left(U^{N_+}_{:, i}\right)^{\transpose}\calE'\right\|^2_2 + \sum_{i \leq \ell}\left\| \left(U^{N_+}_{:, i}\right)^{\transpose}\hat U_{:, k_2+1:d_2} \hat U^{\transpose}_{:, k_2+1:d_2}(N_+ + \calE')\right\|^2_2\right\} \\
& \leq  O\left(\ell \|\calE'_L\|^2_2 + \|\calE'_T\|^2_F + \sum_{i \leq \ell} \left\|\left(U^{N_+}_{:, i}\right)^{\transpose}\right\|^2_2 \left\|\hat U_{:, k_2+1:d_2}\right\|^2_2 \underbrace{\left\| \hat U^{\transpose}_{:, k_2+1:d_2}(N_+ + \calE')\right\|^2_2}_{\leq \frac{\theta^2 \sigma^2_{\epsilon}d_2}{n} \mbox{\scriptsize by the definition of $k_2$.}}\right) \\
& = O\left(\ell\|\calE'_L\|^2_2 + \|\calE'_T\|^2_F + \frac{\ell d_2 \theta^2 \sigma^2_{\epsilon}}{n}\right) \\
& = O\left(\frac{\ell d_2 \theta^2 \sigma^2_{\epsilon}}{n} + + \|\calE'_T\|^2_F\right)
\end{align*}

\myparab{Term 2. $\sum_{i \leq s} \| \hat \mmu_i^{\transpose}\left(\mP_{k_2}(N_+ + \calE') - N_+\right)\|^2_2$.}
 We have 
\begin{align*}
& \sum_{i \leq s} \| \hat \mmu_i^{\transpose}\left(\mP_{k_2}(N_+ + \calE') - N_+\right)\|^2_2 = \sum_{i \leq s} \| \hat \mmu_i^{\transpose}\left(\mP_{k_2}(N_+ + \calE') - (N_+ + \calE') + \calE'\right)\|^2_2 \\
& = \sum_{i \leq s}\| \hat \mmu^{\transpose}_i \calE'\|^2_2 
\end{align*}

On the other hand, note that 
\begin{align}
& \sum_{i \leq s}\|\hat \mmu^{\transpose}_i N_+\|^2_2 \nonumber \\
& = \sum_{i \leq s}\left\|\hat \mmu^{\transpose}_i (N_+ + \calE') - \hat \mmu^{\transpose}_i \calE'\right\|^2_2  \nonumber \\
& = \sum_{i \leq s}\left(\left\|\hat \mmu^{\transpose}_i (N_+ + \calE')\right\|^2_2 + \|\hat \mmu^{\transpose}_i \calE'\|^2_2 - 2 \left\langle \hat \mmu^{\transpose}_i (N_+ +\calE'), \hat \mmu^{\transpose}_i(\calE'_L + \calE'_T) \right\rangle\right) \nonumber \\
& = \sum_{i \leq s} \left(\left\|\hat \mmu^{\transpose}_i (N_+ + \calE')\right\|^2_2 - 2 \langle \hat \mmu^{\transpose}_i(N_++\calE'), \hat \mmu^{\transpose}_i \calE'_L\rangle\right) + \underbrace{\sum_{i \leq s}\|\hat \mmu^{\transpose}_i \calE'\|^2_2}_{\mbox{\footnotesize $=$ Term 2.}} - 2 \sum_{i \leq s}\langle \hat \mmu^{\transpose}_i(N_+ + \calE'), \hat \mmu^{\transpose}_i \calE'_{T}\rangle \label{eqn:geometricineq}
\end{align}

Note that 
\begin{align}
& \sum_{i \leq s} \left(\left\|\hat \mmu^{\transpose}_i (N_+ + \calE')\right\|^2_2 - 2 \langle \hat \mmu^{\transpose}_i(N_++\calE'), \hat \mmu^{\transpose}_i \calE'_L\rangle\right) \nonumber \\
& \geq \sum_{i \leq s}\|\hat \mmu^{\transpose}_i(N_+ + \calE')\|_2\big(\underbrace{\|\hat \mmu^{\transpose}_i(N_+ + \calE')\|_2}_{\geq \theta \sigma_{\epsilon}\sqrt{\frac{d_2}n}} - 2\underbrace{ \|\hat \mmu^{\transpose}_i \calE'_L\|_2}_{\leq 2.2 \sigma_{\epsilon}\sqrt{\frac{d_2}{n}}}\big) \nonumber \\
& \geq 0 \mbox{  (using the fact that $\theta$ is sufficiently large).} \label{eqn:csgzero}
\end{align}

Next, we examine the term $- 2 \sum_{i \leq s}\langle \hat \mmu^{\transpose}_i(N_+ + \calE'), \hat \mmu^{\transpose}_i \calE'_{T}\rangle$. 

\begin{align}
& - 2 \sum_{i \leq s}\langle \hat \mmu^{\transpose}_i(N_+ + \calE'), \hat \mmu^{\transpose}_i \calE'_{T}\rangle \nonumber \\
& = -2 \langle \hat U^{\transpose}_{[s]}(N_+ + \calE'), \hat U^{\transpose}_{[s]}\calE'_T\rangle  \nonumber\\
& = - 2 \trace\left(\hat U^{\transpose}_{[s]}\calE'_T (N_+ + \calE')^{\transpose}\hat U_{[s]}\right)  \nonumber\\
& \geq -2 \left\|\trace(\calE'_T(N_++\calE')^{\transpose}\right\|_2\left\|\hat U^{\transpose}_{[s]}\hat U_{[s]}\right\|_2  \nonumber\\
& = -2 \left|\langle (\calE'_T)^{\transpose}, N_+ + \calE'\rangle\right|  \nonumber \\
& \geq - 2\|\calE'_T\|_F \|N_+ + \calE'\|_F & \mbox{(Cauchy Schwarz)} \nonumber\\ 
& \geq -2 \|\calE'_T\|_F(\|N_+\|_F + \|\calE'\|_F))   \nonumber\\
& \geq -O(\|\calE'_T\|_F)& \mbox{\footnotesize ($\|N_+\|^2_F \leq \|N\|^2_F = \E[\|N\mz\|^2] = \E[\|M\mx\|^2] = O(1))$}\label{eqn:traceeq}
\end{align}

(\ref{eqn:geometricineq}), (\ref{eqn:csgzero}), and (\ref{eqn:traceeq})
 imply that 
$$\sum_{i \leq s} \left\|\mP_{k_2}(N_+ + \calE') - N_+\right\|^2_2 \leq \sum_{i \leq s}\left\|\hat U^{\transpose}N_+\right\|^2_2 + O(\|\calE'_T\|_F).$$

\myparab{Term 3.} We have
$$\sum_{i \leq r}\| \mr^{\transpose}_i \left(\mP_{k_2}(N_+ + \calE') - N_+\right)\|^2_2 = \sum_{i \leq r}\| \mr^{\transpose}_i N_+\|^2_2.$$
This is because $\mr_i$'s are orthogonal to the first $k_2$ left singular vectors of $N_++\calE'$.   

We sum together all the terms:
\begin{align}
    \left\|\mP_{k_2}(N_+ + \calE') - N_+\right\|^2_F & \leq O\left(\frac{\ell d_2\theta^2 \sigma^2_{\epsilon}}{n}\right) + \underbrace{\sum_{i \leq s}\left\|\hat U^{\transpose}_i N_+\right\|_2^2}_{(*)} + \underbrace{\sum_{i \leq t}\|\mr^{\transpose}_i N_+\|^2_2}_{(**)} + O\left(\|\calE'_T\|_F\right). \nonumber \\
    & = O\left(\frac{\ell d_2\theta^2 \sigma^2_{\epsilon}}{n}\right) + \|N_+\|^2_F - \sum_{i \leq \ell}\left(\sigma^{N_+}_i\right)^2 +  O\left(\|\calE'_T\|_F\right) \label{eqn:boundsum}
\end{align}
In the above analysis, we used $\spans(\{\hat \mmu_i\}_{i \leq s}, \{\mr_i\}_{i \leq t})$ is orthogonal to $\spans(\{U^{N_+}_{:,i}\}_{i \leq \ell})$. Therefore, we can collect (*) and (**) and obtain
$$\sum_{i \leq s}\left\|\hat U^{\transpose}_i N_+\right\|_2^2 + \sum_{i \leq t}\|\mr^{\transpose}_i N_+\|^2_2 = \|N_+\|^2_F - \sum_{i \leq \ell}(\sigma^{N_+}_i)^2.$$

Now the MSE is in terms of $\sigma^{N_+}_i$. We aim to bound the MSE in $\sigma^N_i$. So we next relate $\sum_{i \leq \ell}(\sigma^{N_+}_i)^2$ with 
$\sum_{i \leq \ell}(\sigma^{N}_i)^2$. Recall that $\tilde N_+ = [N_+, \mzero]$, where $\mzero \in \reals^{d_2 \times (d_1 - k_1)}$. The singular values of $\tilde N_+$ are the same as those of $N_+$. By using a standard matrix perturbation result, we have 

\begin{equation}\label{eqn:compare}
\sum_{i \leq \ell}\left(\sigma^{N_+}_i - \sigma^N_i\right)^2 = \sum_{i \leq \ell}\left(\sigma^{\tilde N_+}_i - \sigma^N_i\right)^2 \leq \|N_-\|^2_F = c\delta^{\frac{\omega - 1}{\omega+1}}
\end{equation}
 for some constant $c$. We may think (\ref{eqn:compare}) as a constraint and maximize the difference $\sum_{i \leq \ell}(\sigma^N_i)^2 - \sum_{i \leq \ell}(\sigma^{N_+}_i)^2$. This is maximized when $\sigma^N_1 = \sigma^{N_+}_1 + \sqrt{c\delta^{\frac{\omega - 1}{\omega+1}}}$ and $\sigma^N_i = \sigma^{N_+}_i$ for $i > 1$.

Therefore, 
\begin{align}
\sum_{i \leq \ell}\left(\sigma^N_i\right)^2 & \leq \sum_{1 \leq i \leq \ell}\left(\sigma^{N_+}_i\right)^2 + \left(\sigma^{N_+}_1 + \sqrt{c\delta^{\frac{\omega - 1}{\omega + 1}}}\right)^2  \nonumber \\\
& =  \sum_{1 \leq i \leq \ell}\left(\sigma^{N_+}_i\right)^2 + O\left(\sqrt{\delta^{\frac{\omega - 1}{\omega+1}}}\right). \label{eqn:eigencompare}
\end{align} 
 
Now (\ref{eqn:boundsum}) becomes
\begin{equation}\label{eqn:transformedsum}
    (\ref{eqn:boundsum}) \leq \|N_+\|^2_F - \sum_{i \leq \ell}(\sigma^N_i)^2 + O\left(\frac{\ell d_2 \theta^2\sigma^2_{\epsilon}}{n}\right) + O(\xi/\delta^3) + O\left(\sqrt{\delta^{\frac{\omega - 1}{\omega + 1}}}\right). 
\end{equation}
 
Next, we assert that $\ell \geq \ell^*$. Recall that $\|\sigma_i^{N_+} - \sigma_i^N\|^2_2 = \|N_-\|^2_F = o(1)$. This implies $\sigma_i^{N_+} > \theta \sigma_{\epsilon}\sqrt{\frac{d_2}{n}}$ for $i \leq \ell^*$, i.e., $\ell \geq \ell^*$. So we have 
\begin{equation}
\Phi_1 = \|\mP_{k_2}(N_++\calE') - N_+\|^2_F \leq \|N\|^2_F - \sum_{i \leq \ell^*}(\sigma^N_i)^2 + O\left(\frac{\ell^* d_2 \theta^2\sigma^2_{\epsilon}}{n}\right) + O(\xi/\delta^3) + O\left(\sqrt{\delta^{\frac{\omega - 1}{\omega + 1}}}\right). 
\end{equation}
 
Finally, we obtain the bound for $\Phi_1 + \Phi_2 + 2 \sqrt{\Phi_1 \Phi_2}$. 
Note that 
$$\frac{\Phi_2}{\Phi_1} \leq \frac{O\left(\frac{\xi^2}{\delta^6} + \delta^{\frac{\omega - 1}{\omega + 1}}\right)}{O\left(\frac{\xi}{\delta^3} + \sqrt{\delta^{\frac{\omega - 1}{\omega + 1}}}\right)} \leq \min \left\{\frac{\xi}{\delta^3}, \sqrt{\delta^{\frac{\omega - 1}{\omega + 1}}}\right\} (=o(1)).$$

We have 
\begin{align}
    & \Phi_1 + \Phi_2 +  2 \sqrt{\Phi_1 \Phi_2} \nonumber\\
    & = \Phi_1 (1+2\sqrt{\frac{\Phi_2}{\Phi_1}}) + \Phi_2 \nonumber \\
    & = \left[\|N\|^2_F - \sum_{i \leq \ell^*}(\sigma^N_i)^2 + O\left(\frac{\ell^* d_2 \theta^2\sigma^2_{\epsilon}}{n}\right) + O(\xi/\delta^3) + O\left(\sqrt{\delta^{\frac{\omega - 1}{\omega + 1}}}\right)\right]  \nonumber\\
    & \quad \quad \times \left\{1 + \min\left\{\sqrt{\frac{\xi}{\delta^3}} + \left(\delta^{\frac{\omega - 1}{\omega + 1}}\right)^{\frac 1 4}\right\}\right\} + O\left(\frac{\xi^2}{\delta^6}\right) + O\left(\delta^{\frac{\omega - 1}{\omega + 1}}\right) \nonumber\\
    & \leq {\|N\|^2_F - \sum_{i \leq \ell^*}(\sigma^N_i)^2}_{\mbox{}} + {O\left(\frac{\ell^* d_2 \theta^2 \sigma^2_{\epsilon}}{n}\right)}_{\mbox{}} + {O\left(\sqrt{\frac{\xi}{\delta^3}}\right)}_{\mbox{}} + {O\left(\delta^{\frac{\omega - 1}{4(\omega+1)}}\right)}_{\mbox{}}. \label{eqn:terms}
\end{align}

This completes the proof of Theorem~\ref{thm:main_upper}.

\fi
\ifnum\full=0

\fi

\section{Lower bound}\label{sec:lowerbound}

\ifnum\full=0

\fi

\begin{figure}[th!]
\vspace{-2mm}
\centering
\captionsetup{width=1.0\linewidth}
\includegraphics[width=\linewidth]{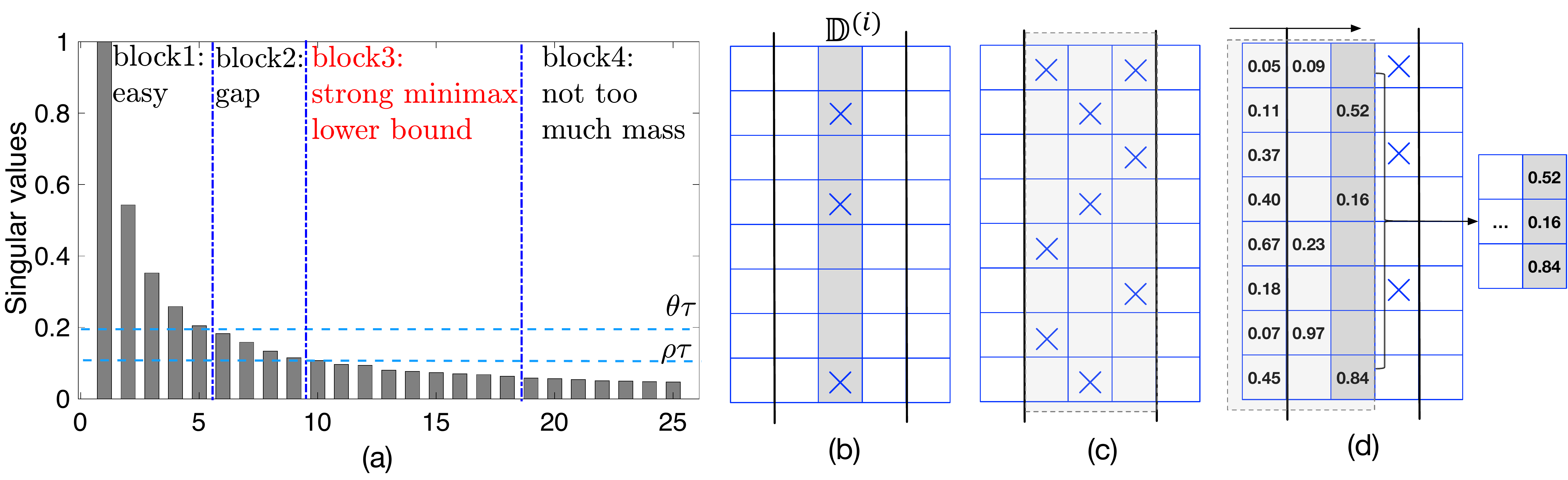} 
\vspace{-2mm}
\caption{{\footnotesize (a) \textbf{Major result:} signals in $N$ are partitioned into four blocks. All signals in block 1 can be estimated (Thm~\ref{thm:main_upper}). All signals in block 3 cannot be estimated (Prop~\ref{prop:lowerbound}). Our lower bound techniques does not handle a small tail in Block 4. A gap in block 2 exists between upper and lower bounds.  (b)-(d) \textbf{Constructing $\bbN$:}} {Step 1 and 2 belong to the first stage; step 3 belongs to the second stage. (b) Step 1. Generate a random subset $\bbD^{(i)}$ for each row $i$, representing its non-zero positions. (c) Step 2. Randomly sample from $\bbD$, where $\bbD$ is the Cartesian product of $\bbD^{(i)}$. (d) Step 3. Fill in non-zero entries sequentially from left to right.}}
\label{fig:pack_U}
\end{figure}
Our algorithm accurately estimates the singular vectors of $N$ that correspond
to singular values above the threshold $\tau = \theta
\sigma_{\epsilon}{\sqrt{\frac{d_2}{n}}}$. However, it may well happen that most of the spectral `mass' of $N$ lies only slightly below this threshold $\tau$. In this section, we establish that \emph{no algorithm} can do better than us, in a bi-criteria sense, i.e. we show that any algorithm that has a slightly smaller sample than ours can only minimally outperform ours in terms of MSE.

We establish `instance dependent' lower bounds: When there is more `spectral
mass' below the threshold, the performance of our algorithm will be worse, and
we will need to establish that no algorithm can do much better. This departs
from the standard minimax framework, in which one examines the entire parameter
space of $N$, e.g. all rank $r$ matrices, and produces a large set of
statistically indistinguishable `bad'
instances~\cite{tsybakov2008introduction}. These lower bounds are not sensitive
to instance-specific quantities such as the spectrum of $N$, and in particular, if prior
knowledge suggests that the unknown parameter $N$ is far from these bad instances,
the minimax lower bound cannot be applied. 

We introduce the notion of \emph{local minimax}. We partition the space into
parts so that \emph{similar} matrices are together.
Similar matrices are those $N$ that have the same singular values and right
singular vectors; we establish strong lower bounds even against algorithms that
know the singular values and right singular vectors of $N$.  An equivalent view
is to assume that the algorithm has oracle access to $C^*$, $M$'s singular
values, and $M$'s right singular vectors. This algorithm can solve
the orthogonalized form as $N$'s singular values and right singular vectors can
easily be deduced. Thus, the only reason why the algorithm needs data is to
\emph{learn} the left singular vectors of $N$. The lower bound we establish is the minimax
bound for this `unfair' comparison, where the competing algorithm is given more
information. In fact, this can be reduced further, i.e., even if
the algorithm `knows' that the left singular vectors of $N$ are sparse,
identifying the locations of the non-zero entries is the key difficulty that
leads to the lower bound.


\begin{definition}[Local minimax bound] Consider a model $\my = M \mx +
	\epsilon$, where $\vx$ is a random vector, so $C^*(\vx) = \E[\mx
	\mx^{\transpose}]$ represents the co-variance matrix of the data
	distribution, and $M = U^M \Sigma^M (V^{M})^{\transpose}$. The relation $(M,
	\vx) \sim (M',\vx') \Leftrightarrow (\Sigma^M = \Sigma^{M'} \wedge V^M =
	V^{M'} \wedge C^*(\vx) = C^*(\vx'))$ is an equivalence relation and let the
	equivalence class of $(M, \vx)$ be 
\ifnum\full=0
$
\calR(M, \mx) = \{(M', \mx'): \Sigma^{M'} = \Sigma^M, V^{M'} = V^M,   \mbox{ and } C^*(\mx') = C^*(\mx)\}. 
$
\else
\begin{equation}
	\calR(M, \mx) = \{(M', \mx'): \Sigma^{M'} = \Sigma^M, V^{M'} = V^M, \mbox{ and } C^*(\mx') = C^*(\mx)\}. 
\end{equation}
\fi
The local minimax bound for $\my = M \mx + \epsilon$ with $n$ independent samples and $\epsilon \sim N(0, \sigma^2_{\epsilon}I_{d_2\times d_2})$ is 
\ifnum\full=0
\vspace{-2mm}
\begin{equation}
\mr(\mx, M, n, \sigma_{\epsilon}) = \min_{\hat M} \max_{(M', \mx') \in \calR(M, \mx)} E_{\tiny{\substack{\mbox{$\mX, \mY$ from } \\ \my \sim M' \mx' + \epsilon}}} \Big[\E_{\mx'}[\|\hat M(\mX, \mY) \mx' - M' \mx'\|^2_2\mid \mX, \mY]\Big].
\label{eqn:specialminimax}
\end{equation}
\else
\begin{equation}\label{eqn:specialminimax}
\mr(\mx, M, n, \sigma_{\epsilon}) = \min_{\hat M} \max_{(M', \mx') \in \calR(M, \mx)} \E_{\tiny{\substack{\mbox{$\mX, \mY$ from } \\ \my \sim M' \mx' + \epsilon}}}[\E_{\mx'}[\|\hat M(\mX, \mY) \mx' - M' \mx'\|^2_2\mid \mX, \mY]]. 
\end{equation}
\fi
\end{definition} 
\vspace{-2mm}
It is worth interpreting \eqref{eqn:specialminimax} in some detail. For any
two $(M, \vx)$, $(M', \vx')$ in $\calR(M, \vx)$, the algorithm has the same
`prior knowledge', so it can only distinguish between the two
instances by using the \emph{observed data}, in particular $\hat M$ is a 
function only of $\mX$ and $\mY$, and we denote it as $\hat M(\mX, \mY)$ to
emphasize this. Thus, we can evaluate the performance of $\hat{M}$ by looking at
the worst possible $(M', \vx')$ and considering the MSE $\E \Vert \hat{M} (\mX,
\mY) \vx' - M' \vx'\Vert^2$. 


\begin{proposition}\label{prop:lowerbound} Consider the problem $\my = M\mx +\epsilon$ with normalized form $\my = N\mz + \epsilon$. 
Let $\xi$ be a sufficient small constant. There exists a sufficiently small constant $\rho_0$ (that depends on $\xi$) and a constant $c$ such that for any $\rho \leq \rho_0$, 
$
\mr(\mx, M, n, \sigma_{\epsilon}) \geq (1-c\rho^{\frac 1 2 - \xi}) \sum_{i \geq \lt} (\sigma^{N}_i)^2 -
O\Bigl(\frac{\rho^{\frac 1 2 - \xi }}{d_2^{\omega - 1}}\Bigr), 
$
where $\lt$ is the smallest index such that $\sigma^{N}_{\lt} \leq \rho \sigma_{\epsilon}\sqrt{\frac{d_2}{n}}$.
\end{proposition}
\vspace{-2mm}

Proposition~\ref{prop:lowerbound} gives the lower bound on the MSE in
expectation; it can be turned into a high probability result with suitable
modifications. The proof of the lower bound uses a similar 
`trick' to the one used in the analysis of the upper bound analysis to cut the tail. This results in an additional term
$O\Bigl(\frac{\rho^{\frac 1 2 - \xi }}{d_2^{\omega - 1}}\Bigr)$
which is generally smaller than the
$n^{-c_0}$ tail term in Theorem~\ref{thm:main_upper} and 
does not dominate the gap. 

\myparab{Gap requirement and bi-criteria approximation algorithms.} 
Let $\tau =  \sigma_{\epsilon}\sqrt{\frac{d_2}{n}}$.
Theorem~\ref{thm:main_upper} asserts that any signal above the threshold
$\theta\tau$ can be detected, i.e., the MSE is at most $\sum_{\sigma^N_i >
\theta \tau}\sigma^2_i(N)$ (plus inevitable noise), whereas
Proposition~\ref{prop:lowerbound} asserts that any signal below the threshold
$\rho \tau$ cannot be detected, i.e., the MSE is approximately at least
$\sum_{\sigma^N_i \geq \rho \tau}(1-\mathrm{poly}(\rho))\sigma^2_i(N)$. There
is a `gap' between $\theta \tau$ and $\rho \tau$, as $\theta > 1$ and $\rho <
1$. See Fig.~\ref{fig:pack_U}(a). This kind of gap is inevitable because both bounds are `high probability'
statements. This gap phenomenon appears naturally when the sample size is small
as can be illustrated by this simple example. Consider the problem of
estimating $\mu$ when we see one sample from $N(\mu, \sigma^2)$.  Roughly
speaking, when $\mu \gg \sigma$, the estimation is feasible, and whereas $\mu \ll
\sigma$, the estimation is impossible. For the region $\mu \approx \sigma$,
algorithms fail with constant probability and we cannot prove a high
probability lower bound either. 

While many of the signals can `hide' in the gap, the inability to detect signals in the gap is a transient phenomenon. When the number of samples $n$ is
modestly increased, our detection threshold $\tau = \theta
\sigma_{\epsilon}\sqrt{\frac{d_2}{n}}$ shrinks, and this hidden signal can be
fully recovered. This observation naturally leads to a notion of bi-criteria
optimization that frequently arises in approximation algorithms. 


\begin{definition} An algorithm for solving the $\my = M \mx + \epsilon$ problem is $(\alpha, \beta)$-optimal if, when given an i.i.d. sample of size $\alpha n$ as input, it outputs an estimator whose MSE is at most 
$\beta$ worse than the local minimax bound, i.e., 
$\E[\|\hat \my - \my \|^2_2] \leq \mr(\mx, M, n, \sigma_{\epsilon}) + \beta.$

\end{definition}

\begin{corollary}\label{cor:approx} Let $\xi$ and $c_0$ be  small constants and $\rho$ be a tunable parameter. Our algorithm is $(\alpha, \beta)$-optimal for \ifnum\full=0
$\alpha  = \frac{\theta^2}{\rho^{\frac 5 2}}$ and $\beta  = O(\rho^{\frac 1 2 - \xi})\|M\mx\|^2_2 + O(n^{-c_0})$
\else
\begin{align*}
\alpha & = \frac{\theta^2}{\rho^{\frac 5 2}} \\
\beta & = O(\rho^{\frac 1 2 - \xi})\|M\mx\|^2_2 + O(n^{-c_0}) 
\end{align*}
\fi
\end{corollary}
\vspace{-4mm}

The error term $\beta$ consists of $\rho^{\frac 1 2-\epsilon}\|M\mx\|^2_2$ that
is directly characterized by the signal strength and an additive term
$O(n^{-c_0}) = o(1)$. Assuming that $\|M\mx\| = \Omega(1)$, i.e., the signal 
is not too weak, the term $\beta$ becomes a single multiplicative
bound
$O(\rho^{\frac 1 2 - \xi} + n^{-c_0})\|M\mx\|^2_2$.
This gives an easily interpretable result. For example, when our data size is
$n\log n$, the performance gap between our algorithm and \emph{any} algorithm
that uses $n$ samples is at most $o(\|M\mx\|^2_2)$. The improvement is
significant when other baselines deliver MSE in the additive form that could be
larger than $\|M\mx\|^2_2$ in the regime $n \leq d_1$.

\myparab{Preview of techniques.} Let $N = U^N \Sigma^N (V^N)^{\transpose}$ be the instance (in orthogonalized form). Our goal is to construct a collection $\calN = \{N_1, \dots, N_K\}$ of $K$ matrices so that \emph{(i)} For any $N_i \in \calN$, $\Sigma^{N_i} = \Sigma^N$ and $V^{N_i} = V^N$. \emph{(ii)} For any two $N_i, N_j \in \calN$, $\|N - N'\|_F$ is large, and \emph{(iii)} $K = \exp(\Omega(\mathrm{poly}(\rho)d^2))$ (cf.\ifnum\full=0~\citep[Chap. 2]{tsybakov2008introduction})\else~\cite[Chap. 2]{tsybakov2008introduction}).\fi

Condition (i) ensures that it suffices to construct unitary matrices
$U^{N_i}$'s for $\calN$, and that the resulting instances will be in the same
equivalence class.  Conditions (ii) and (iii) resemble standard construction of
codes in information theory: we need a large `code rate', corresponding to
requiring a large $K$ as well as large distances between codewords,
corresponding to requiring that $\|U_i - U_j\|_F$ be large. Standard approaches
for constructing such collections run into difficulties. Getting a sufficiently
tight concentration bound on the distance between two random unitary matrices
is difficult as the matrix entries, by necessity, are correlated.
On the other hand, starting with a large collection of random unit vectors and
using its Cartesian product to build matrices does not  necessarily yield
unitary matrices.

We design a two-stage approach to decouple condition (iii) from (i) and (ii) by only generating sparse matrices $U^{N_i}$. See Fig.~\ref{fig:pack_U}(b)-(d). In the first stage (Steps 1 \& 2 in Fig.~\ref{fig:pack_U}(b)-(c)), we only specify the
non-zero positions (sparsity pattern) in each $U^{N_i}$. It suffices to
guarantee that the sparsity patterns of the matrices $U^{N_i}$ and $U^{N_j}$
have little overlap. The existence of such objects can easily be proved using
the probabilistic method.  Thus, in the first stage, we can build up a large
number of sparsity patterns. In the second stage (Step 3 in Fig.~\ref{fig:pack_U}(d)), we carefully fill in values
in the non-zero positions for each $U^{N_i}$. When the number of non-zero
entries is not too small, satisfying the unitary constraint is feasible. As the
overlap of sparsity patterns of any two matrices is small, we can argue the
distance between them is large. By carefully trading off the number of non-zero
positions and the portion of overlap, we can simultaneously satisfy all three
conditions. 

\ifnum\full=1


\subsection{Roadmap}
This section describes the roadmap for executing the above idea. 

\myparab{Normalized form.} Recall that $d_2 \leq d_1$, we have
\begin{equation}
M \mx = \underbrace{U^M}_{d_2 \times d_2} \underbrace{\Sigma^{M}}_{d_2 \times d_2} \underbrace{(V^{M})^{\transpose}}_{d_2 \times d_1} \underbrace{V^*}_{d_1 \times d_1}\underbrace{ (\Lambda^*)^{\frac 1 2}}_{d_1 \times d_1}\underbrace{\mz}_{d_1 \times 1}
\end{equation}

We may perform an SVD on $\Sigma^M (V^M)^{\transpose}V^* (\Lambda^*)^{\frac 1 2} = \underbrace{A}_{d_2 \times d_2} \underbrace{L^{\normf}}_{d_2 \times d_2} \underbrace{B^{\transpose}}_{d_2 \times d_1}$. We may also set $\mz^{\normf} = B^{\transpose}\mz$, which is a standard multi-variate Gaussian in $\reals^{d_2}$. Then we have 
\begin{equation}
    M \mx = U^M A L^{\normf}B^{\transpose} \mz = (U^M A) L^{\normf} \mz^{\normf}.
\end{equation}
Let $N^{\normf} = (U^M A) L^{\normf}$. The SVD of $N^{\normf}$ is exactly $(U^M)AL^{\normf} I_{d_2 \times d_2}$ because $U^M A$ is unitary. 
The \emph{normalized form} of our problem is 
\begin{equation}
\my = N^{\normf} \mz^{\normf} + \epsilon.
\end{equation}

Recall our local minimax has an oracle access interpretation. An algorithm with the oracle can reduce a problem into the normalized form on its own, and the algorithm knows $L^{\normf}$. But the oracle still does not have any information on $N^{\normf}$'s left singular vectors because being able to manipulate $U^M$ is the same as being able to manipulate $U^M A$. Therefore, we can analyze the lower bound in normalized form, with the assumption that the SVD of $N^{\normf} = U^{\normf} L^{\normf}I_{d_2\times d_2}$, in which $L^{\normf}$ is known to the algorithm. We shall also let $\sigma^{\normf}_i = L^{\normf}_{i,i} = \sigma^{N^{\normf}}_i$. Because $N^{\normf}$ is square, we let $d = d_2$ in the analysis below. 

We make two remarks in comparison to the orthogonalized form $\my = N \mz + \epsilon$. (i)  $\mz^{\normf} \in \reals^{d_2}$, whereas $\mz \in \reals^{d_1}$. $\mz^{\normf}$'s dimension is smaller because the knowledge of $\Sigma^M$ and $V^M$ enable us to remove the directions from $\mx$ that are orthogonal to $M$'s row space. \emph{(ii)} $\sigma^N_i = \sigma^{N^{\normf}}_i$ for $i \leq d_2$. 


We rely on the following theorem (Chapter 2 in~\cite{tsybakov2008introduction}) to construct the lower bound. 

\begin{theorem}\label{thm:block} Let $\bbN^{\normf} = \{N^{\normf}_1, N^{\normf}_2, \dots, N^{\normf}_K\}$, where $K \geq 2$. Let $P_i$ be distribution of the training data produced from the model $\my = N^{\normf}_i \mz^{\normf} + \epsilon$. 
Assume that 
\begin{itemize}
\item $\|N^{\normf}_i - N^{\normf}_j\|^2_F  \geq 2s > 0$ for any $0 \leq j \leq k \leq K$. 
\item For any $j = 1, \dots K$ and 
\begin{equation}
\frac 1 K \sum_{j = 1}^K \Kld(P_j, P_1) \leq \alpha \log K
\end{equation}
with $0 \leq \alpha \leq \frac 1 8$.
\end{itemize}
Then 
\begin{equation}
\inf_{\hat N^{\normf}} \sup_{N^{\normf} \in \bbN^{\normf}}
 \Pr_{N^{\normf}}(\|\hat N^{\normf}, N^{\normf}\| \geq s) \geq \frac{\sqrt{K}}{1+\sqrt K}(1-2\alpha - \sqrt{\frac{2\alpha}{\log K}}).  
\end{equation}
\end{theorem}

Because $L^{\normf}$ is fixed, this problem boils down to finding a collection $\bbU^{\normf}$ of unitary matrices in $\reals^{d\times d}$ such that any two elements in $\bbU^{\normf}$ are sufficiently far. Then we can construct $\bbN^{\normf} = \{U^{\normf}L^{\normf}: U^{\normf} \in \bbU\}$. 
 
We next reiterate (with elaboration) the challenge we face when using existing techniques. Then we describe our approach. 
We shall first examine a simple case, in which we need only design vectors for one column. Then we explain the difficulties of using an existing technique to generalize the design. Finally, we explain our solution. 

Recall that $(\mZ^{\normf})^{\transpose} \mY = (N^{\normf})^{\transpose} + \calE$, where we can roughly view $\calE$ as a matrix that consists of independent Gaussian $(0, \sigma_{\epsilon}/\sqrt{n})$. 
For the sake of discussion, we assume $\sigma_{\epsilon} = 1$ in our discussion below.

\myparab{Warmup: one colume case.} The problem of packing one column roughly corresponds to establishing a lower bound on the estimation problem 
$\my = \mmu + \epsilon$, where $\my, \mmu, \epsilon \in \reals^d$.  $\epsilon$ corresponds to a column in $\calE$ and consists of $d$ independent Gaussian $N(0, 1/\sqrt{n})$. 
$\mmu$ correpsonds to a column in $U^{\normf}$ and we require $\|\mmu\|_2 \approx \rho\sqrt{d/n}$. To apply Theorem~\ref{thm:block}, we shall
construct a $\bbD = \{\mmu_1, \dots, \mmu_K\}$ such that $\|\mmu_i - \mmu_j\|$ is large for each $\{i, j\}$ pair and $K$ is also large. Specifically, we require $\|\mmu_i - \mmu_j\|^2_2 \approx 2 \rho^2 \frac d n$ (large distance requirement) and 
$K = \exp(\Theta(\sqrt{\rho} d))$ (large set requirement). $\sqrt \rho$ is carefully optimized and we will defer the reasoning to the full analysis below. 
A standard tool to construct $\bbD$ is to use a probabilistic method. We sample $\mmu_i$ independently from the same distribution and argue that with high probability $\|\mmu_i - \mmu_j\|^2_2$ is sufficiently large. 
Then a union bound can be used to derive $K$. For example, we may set $\mmu_i \sim N(0, \rho\sqrt{\frac{d}{n}}I_{d\times d})$ for all $i$, and the concentration quality suffices for us to sample $K$ vectors.  


\myparab{Multiple column case.} We now move to the problem of packing multiple columns in $U$ together. $K$ is required to be much larger, e.g., $K = \exp(\Theta(\sqrt{\rho} d^2))$ 
for certain problem instances. A natural generalization of one-column case is to build our parameter set by taking the Cartesian product of multiple copies of $\bbD$. This gives us a large $K$ for free but the key issue is 
that vectors in $\bbD$ are generated independently. So there is no way to guarantee they are independent to each other. In fact, it is straightforward to show that many elements in the Cartesian product are far from unitary. 
One may also directly sample random unitary matrices and argue that they are far from each other. But there seems to exist no tool that enables us to build up a concentration of $\exp(-\Theta(\sqrt{\rho} d^2))$ 
between two random unitary matrices. 

Therefore, a fundamental problem is that we need independence to build up a large set but the unitary constraint limits the independence. So we either cannot satisfy
the unitary requirement (Cartesian product approach) or cannot properly use the independence to build concentrations (random unitary matrix approach). 

\myparab{Our approach.} We develop a technique that decouples the three requirements (unitary matrices, large distance, and large $K$). Let us re-examine the Cartesian product approach. When the vectors for each column are completely determined, then it is remarkably difficult to build a Cartesian product that guarantees 
orthogonality between columns. To address this issue, our approach only ``partially'' specify vectors in each column. Then we take a Cartesian product of these partial specifications. 
So the size of the Cartesian product is sufficiently large; meanwhile an element in the product does not fully specify $U^{\normf}$ so we still have room to make them unitary. 
Thus, our final step is to transform each partial specification in the Cartesian product into a unitary matrix. Specifically, it consists of three steps (See Fig.~\ref{fig:pack_U})

\myparab{Step 1. Partial specification of each column.} For each column $i$ of interest, we build up a collection $\bbD^{(i)} = \{R^{(i, 1)}, \dots, R^{(i, K)}\}$.  Each $R^{(i,j)} \subset [d]$ 
specifies only the positions of non-zero entire for a vector prepared for filling in $U_{:, i}$. 

\myparab{Step 2. Cartesian product.} Then we randomly sample elements from the Cartesian product $\bbD \triangleq \bigotimes_i \bbD^{(i)}$. Each element in the product specifies the non-zero entries of $U^{\normf}$. 
We need to do another random sampling instead of using the full Cartesian product because we need to guarantee that any two matrices share a small number of non-zero entries. 
For example, $(R^{(1, 1)}, R^{(2, 1)}, R^{(3, 1)})$ and $(R^{(1, 1)}, R^{(2, 1)}, R^{(3, 2)})$ are two elements in $\bigotimes_i \bbD^{(i)}$ but they specify two matrices 
with the same locations of non-zero entries for the first two columns. 

\myparab{Step 3. Building up unitary matrices.} Finally, for each element $\vec R \in \bbD$ (that specify positions of non-zero entries), we carefully fill in the values of the non-zero entries so that all our matrices are unitary and far from each other. 
We shall show that it is \emph{always} possible to construct unitary matrices that ``comply with'' $\vec R$. In addition, our unitary matrices have few entries with large magnitude so when two matrices share few positions of non-zero entries, they are far.

\subsection{Analysis}
We now execute the plan outlined in the roadmap. Let $K$ and $\lambda$ be tunable parameters. For the purpose of getting intuition, 
$K$ is related to the size of $\bbU^{\normf}$ so it can be thought as being exponential in $d$,  whereas $\lambda$ is a constant and we use $\rho^{\lambda}$ 
to control the density of non-zero entries in each $U^{\normf} \in \bbU^{\normf}$. 

Let $\bbD^{(i)} = \{R^{(i, 1)}, R^{(i, 2)}, \cdots, R^{(i, K)}\}$ be a collection of random subsets in $[d]$, in which each $R^{(i, j)}$ 
is of size $\rho^{\lambda}d$. We sample the subsets without replacement. Recall  that $\lt$ is the smallest index such that $\sigma^{\normf}_{\lt} \leq\rho \sigma_{\epsilon}\sqrt{\frac d n}$ 
and let us also define $\ut = \lfloor \frac{\rho^{\lambda} d}{2}\rfloor$. We let $\gamma = \ut - \lt + 1$. We assume that $\ut \geq \lt$; otherwise Proposition~\ref{prop:lowerbound} becomes trivial. Let 
$\bbD$ be the Cartesian product of $\bbD^{(i)}$ for integers $i \in [\lt, \ut]$. We use $\vec R$ to denote an element in $\bbD$. $\vec R = (\vec R_{\lt}, \vec R_{\lt + 1}, \cdots, \vec R_{\ut})$ is a 
$\gamma$-tuple so that each element $\vec R_i$ corresponds to an element in $\bbD^{(i)}$. There are two ways to represent $\vec R$. Both are found useful in our analysis. 

\begin{enumerate}
\item \emph{Set representation.} We treat $\vec R_i$ as a set in $\bbD^{(i)}$. 
\item \emph{Index representation.} We treat $\vec R_i$ as an index from $[K]$ that specifies the index of the set that $\vec R_i$ refers to. 
\end{enumerate}

Note that the subscript $i$ of $\vec R_i$ \emph{starts at} $\lt$ (instead of $1$ or $0$) for convenience.

\begin{example}\emph{ The index of $\bbD_i$ starts at $\lt$. Assume that $\ut = \lt + 1$. Let $\bbD^{(\lt)} = \big(\{2,3\}, \{1, 4\}, \{1, 2\}\big)$ and $\bbD^{(\lt + 1)} = \big(\{1, 3\}, \{2, 4\}, \{3, 4\}\big)$. The element $\big(\{1, 2\}, \{2, 4\}\big) \in \bbD^{(\lt)} \otimes \bbD^{(\lt+1)}$. There are two ways to represent this element. \emph{(i) Set representation.} $\vec R = \big(\{1, 2\}, \{2, 4\}\big)$, in which $\vec R_{\lt} = \{1, 2\}$ and $\vec R_{\lt + 1} = \{2, 4\}$.  \emph{(ii) Index representation.} $\vec R = (3, 2)$. $\vec R_{\lt} = 3$ refers to that the third element $\{1, 2\}$ in $\bbD^{(\lt)}$ is selected. 
}
\end{example}

We now describe our proof in detail. We remark that throughout our analysis, constants are re-used in different proofs.

\subsubsection{Step 1. Partial specification for each column}
This step needs only characterize the behavior of an individual $\bbD^{(i)}$.

\begin{lemma}\label{lem:randomsubset} Let $\rho < 1$ be a sufficiently small variable, $\lambda$ be a tunable parameter and let $\bbD^{(i)} = \{R^{(i, 1)}, R^{(i, 2)}, \cdots , R^{(i, K)}\}$ be a collection of random subsets in $[d]$ (sampled without replacement)
such that $|R^{(i, j)}| = \rho^{\lambda}d$ for all $j$. There exist constants $c_0$, $c_1$, and $c_2$ such that when $K = \exp(c_0 \rho^{2\lambda}d)$, with probability $1 - \exp(-c_1\rho^{2\lambda}d)$, for any two distinct
$R^{(i)}$ and $R^{(j)}$, $|R^{(i, j)} \cap R^{(i, k)}| \leq c_2 \rho^{2\lambda}d$. 
\end{lemma}

\begin{proof}[Proof of Lemma~\ref{lem:randomsubset}]
This can be proved by a standard probabilistic argument. 
Let $R^{(i, j)}$ be an arbitrary subset such that $|R^{(i, j)}| = \rho^{\lambda}d$. Let $R^{(i, k)}$ be a random subset of size $\rho^{\lambda}d$. We compute 
the probability that $|R^{(i, j)} \cap R^{(i, k)}| \geq c_2 \rho^{2\lambda}d$ for a fixed $R^{(i, j)}$. 

Let us sequentially sample elements from $R^{(i, k)}$ (without replacement). Let $I_t$ be an indicator random variable that sets to 1 if and only if 
the $t$-th random element in $R^{(i, k)}$ hits an element in $R^{(i, j)}$. We have
\begin{equation}
\Pr[I_t = 1] \leq \frac{\rho^{\lambda} d}{(1-\rho^{\lambda})d} \leq \frac{\rho^{\lambda}}{2}. 
\end{equation}
By using a Chernoff bound, we have 
\begin{equation}
\Pr\left[|\sum_{i = 1}^{\rho^{\lambda}d}I_t| \geq \frac{c_2\rho^{2\lambda}d}{2}\right] \leq \exp(-\Omega(\rho^{2\lambda}d)). 
\end{equation}
By using a union bound, we have 
\begin{equation}
\Pr[\exists i, j: |R^{(i, j)}\cap R^{(i, k)}| \geq \frac{c_2 \rho^{2\lambda}d}{2}] \leq \binom{K}{2} \exp(-\Omega(\rho^{2\lambda}d)) \leq \exp(-\Omega(\rho^{2\lambda}d) + 2\log K).
\end{equation}
Therefore, when we set $K = \exp(c_0\rho^{2\lambda}d)$, the failure probability is $1 - \exp(-\Theta(\rho^{2\lambda}d))$. 
\end{proof}

\subsubsection{Step 2. Random samples from the Cartesian product}
We let $\bbD = \bigotimes_{i \in [\lt, \ut]}\bbD^{(i)}$. Note that each $\bbD^{(i)}$ is sampled independently. We define $\bbS$ be a random subset of  $\bbD$. We next explain the reason we need to sample random subsets. Recall that for each $\vec R \in \bbS$, we aim to construct a unitary matrix $U^{\normf}$ (to be discussed in Step 3) such that the positions of non-zero entries in $U^{\normf}_{:, i}$ are specified by $\vec R_i$ (i.e., $U^{\normf}_{j, i} \neq 0$ only when $j \in \vec R_i$). 

Let $\vec R$ and $\vec R'$ be two distinct elements in $\bbD$. Let $U^{\normf}$ and $\tilde U^{\normf}$ be two unitary matrices generated by $\vec R$ and $\vec R'$. We ultimately aim to have that $\|U^{\normf}L^{\normf} - \tilde U^{\normf} L^{\normf}\|^2_F$ being large. We shall make sure \emph{(i)} $U^{\normf}$ and $\tilde U^{\normf}$ share few non-zero positions (Step 2; this step), and \emph{(ii)} few entries in $U^{\normf}$ and $\tilde U^{\normf}$ have excessive magnitude (Step 3). These two conditions will imply that $U^{\normf}$ and $\tilde U^{\normf}$ are far, which then implies a lower bound on $U^{\normf}L^{\normf}$ and $\tilde U^{\normf} L^{\normf}$. 

Because we do not want $U^{\normf}$ and $\tilde U^{\normf}$ share non-zero positions, we want to maximize the Hamming distance (in index representation) between $\vec R$ and $\vec R'$ (i.e., $\vec R_i = \vec R'_i$ implies $U^{\normf}_{:, i}$ and $\tilde U^{\normf}_{:, i}$ share all non-zero positions, which is a bad event). We sample $\bbS$ randomly from $\bbD$ because random sample is a known procedure that generates ``code'' with large Hamming distance~\cite{mackay2003information}. 

Before proceeding, we note one catch in the above explanation, i.e., different 
columns are of different importance. Specifically, $\|U^{\normf}L^{\normf} - \tilde U^{\normf}L^{\normf}\|^2_F = \sum_{i \in [d]} (\sigma^{\normf}_i)^2 \|U^{\normf}_{:, i} - \tilde U^{\normf}_{:, i}\|^2$. When $\vec R_i$ and $\vec R'_i$ collide for a large $(\sigma^{\normf}_i)^2$, it makes more impact to the Frobenius norm. Thus, we define a weighted cost function that resembles the structure of Hamming distance. 

\begin{equation}
\mc(\vec R, \vec R') = \sum_{i \in [\lt, \ut]}(\sigma^{\normf}_i)^2 I(\vec R_i = \vec R'_i). 
\end{equation}

Note that the direction we need for $\mc(\vec R, \vec R')$ is opposite to Hamming distance. One usually maximizes Hamming distance whereas we need to minimize the weighted cost. 

We need to develop a specialized technique to produce concentration behavior for $\bbS$ because $\mc(\vec R, \vec R')$ is weighted. 

\begin{lemma}\label{lem:hamming} Let $\rho$ and $\lambda$ be the parameters for producting $\bbD^{(i)}$. Let $\zeta < 1$ be a tunable parameter. Let $\bbS$ be a random subset of $\bbD$ of $\bbD \triangleq \bigotimes_{i \in [\lt, \ut]}\bbD^{(i)}$  such that
\begin{equation}
|\bbS| = \exp\left(c_3 \frac{n \rho^{2\lambda+\zeta}}{\rho^2\sigma^2_{\epsilon}}\left(\sum_{i \in [\lt, \ut]}(\sigma^{\normf}_i)^2\right)\right)
\end{equation}
for some constant $c_3$. With high probability at least $1 - \exp(-c_4 \rho^{2\lambda}d)$ ($c_4$ a constant), for any $\vec R$ and $\vec R'$ in $\bbS$, 
$\mc(\vec R, \vec R') \leq \rho^{\zeta}(\sum_{i \in [\lt, \ut]} (\sigma^{\normf}_i)^2)$.
\end{lemma}

\begin{proof}
Let $\Psi = \sum_{i \in [\lt, \ut]}(\sigma^{\normf}_i)^2$. 
Let $\vec R$ and $\vec R'$ be two different random elements in $\bbD$. We shall 
first compute that 
$\Pr\left[\mc(\vec R, \vec R') \geq \rho^{\zeta}\Psi\right]$. Here, we assume that $\vec R$ is an arbitrary fixed element and $\vec R'$ is random. 

Recall that $\sigma^{\normf}_i \in [0, \rho \sigma_{\epsilon}\sqrt{\frac n d}]$ for $i \in [\lt, \ut]$. 
We shall partition $[0, \rho \sigma_{\epsilon}\sqrt{\frac n d}]$ into subintervals and group $\sigma^{\normf}_i$ by these intervals. Let $\calI_t$ be the set of $\sigma^{\normf}_i$ that are in $[2^{-t-1}\rho \sigma_{\epsilon}\sqrt{\frac n d}, 2^{-t}\rho\sigma_{\epsilon} \sqrt{\frac n d}]$ ($t\geq 0$).
Let $T$ be the largest integer such that $\calI_T$ is not empty. 
Let $L_t = |\calI_t|$  and $\ell_t = \sum_{i \in \calI_t}I(\vec R_i = \vec R'_i)$. We call $\{\ell_t\}_{t \leq T}$ the \emph{overlapping coefficients} between $\vec R$ and $\vec R'$. 

Note that $\mc(\vec R, \vec R') \leq \sum_{t \leq T} \ell_t 2^{-2t}\rho^2 \sigma^2_{\epsilon}d/n$. Therefore, a necessary condition for $\mc(\vec R, \vec R') \geq \rho^{\zeta}\Psi$ is 
$\sum_{t \leq T}\frac{\ell_t 2^{-2t}\rho^2\sigma^2_{\epsilon}d}{n} \geq \rho^{\eta}\Psi$. Together with the requirement that $\sum_{t \leq T}\ell_t \geq 1$, we need
\begin{equation}
\sum_{t\leq T}\ell_t \geq \max\left\{\frac{n \rho^{\zeta}\Psi}{d \rho^2 \sigma^2_{\epsilon}}, 1\right\}
\end{equation}

Recall that we assume that $\vec R$ is fixed and $\vec R'$ is random. When $\mc(\vec R, \vec R') \geq \rho^{\zeta}\Psi$, we say $\vec R'$ is bad. 
We next partition all bad $\vec R'$ into sets indexed by $\{\ell_t\}_{t \leq T}$. Let $\bbC(\{\ell_t\}_{t\leq T})$ be all the bad  $\vec R'$ such that the overlapping coefficients between $\vec R$ and $\vec R'$ are
$\{\ell_t\}_{t \leq T}$. We have  

{\small
\begin{align*}
\Pr[\mc(\vec R, \vec R') \geq \rho^{\zeta}\Psi] & = \Pr[\vec R' \mbox{ is bad }] = \Pr\left[\vec R' \in \bigcup_{\{\ell_t\}_{t\leq T}}\bbC(\{\ell_t\}_{t \leq T})\right] = \sum_{k \geq 1}\sum_{{\scriptsize\substack{\mbox{all } \bbC(\{\ell_i\})\\ \mbox{s.t. $\sum_{t}\ell_t = k$}}}}\Pr[\vec R' \in \bbC(\{\ell_t\}_{t \leq T})]
\end{align*}
}

Next also note that 
\begin{align*}
\Pr[\vec R' \in \bbC(\{\ell_t\}_{t \leq T})] \leq \prod_{t \leq T}\binom{L_t}{\ell_t}\left(\frac 1 K\right)^{\sum_{t \leq T}\ell_t}, 
\end{align*}
where $K$ is the size of each $\bbD^{(i)}$. 

The number of possible $\{\ell_i\}_{t \leq T}$ such that $\sum_{t \leq T}\ell_i = k$ is at most $\binom{d+k}{k}$. Therefore, 
\begin{align*}
&  \sum_{k \geq 1}\sum_{{\scriptsize\substack{\mbox{all } \bbC(\{\ell_i\})\\ \mbox{s.t. $\sum_{t}\ell_t = k$}}}} \Pr\left[\vec R' \in \bbC(\{\ell_i\}_{i \leq T})\right] \\
& \leq  \sum_{k \geq 1}\sum_{{\scriptsize\substack{\mbox{all } \bbC(\{\ell_i\})\\ \mbox{s.t. $\sum_{t}\ell_t = k$}}}} \prod_{t \leq T}\binom{L_t}{\ell_t}\left(\frac 1 K\right)^{\ell_t} \\
& \leq  \sum_{k\geq 1} \sum_{{\scriptsize\substack{\mbox{all } \bbC(\{\ell_i\})\\ \mbox{s.t. $\sum_{t}\ell_t = k$}}}}  \prod_{t \leq T}\left(\frac{e L_t}{K}\right)^{\ell_t} \quad \mbox{(using $\binom{L_t}{\ell_t} \leq \left(\frac{eL_t}{\ell_t}\right)^{\ell_t} \leq (eL_i)^{\ell_i}$)} \\
& \leq  \sum_{k \geq 1} \binom{d+k}{k}\prod_{t \leq T}\left(\frac{eL_t}{K}\right)^{\ell_t} \\
& \leq  d \max_{{\scriptsize\substack{k \mbox{ s.t.} \\ \sum_t \ell_t = k}}} \binom{d+k}{k}\prod_{i \leq T}\left(\frac{eL_i}{K}\right)^{\ell_t} \\
& \leq  d \left(\frac{e(d+k)}{k}\right)^{k} \prod_{t \leq T}\left(\frac{eL_t}{K}\right)^{\ell_t} \quad \mbox{(where $k = \sum_{t \leq T}\ell_t$ from the previous line)} \\
& \leq d (2ed)^k\prod_{t \leq T}\left(\frac{eL_t}{K}\right)^{\ell_t} \\
& \leq  d \prod_{t \leq T}\left(\frac{2e^2d^2}{K}\right)^{\ell_t}  \leq \exp\left(-c \rho^{2\lambda}d(\sum_{t \leq T}\ell_t)\right)  \quad \mbox{($c$ is a suitable constant; using $K$ from Lemma~\ref{lem:randomsubset}) }\\
& \leq  \exp\left(-c \rho^{2\lambda}d \max\left\{\frac{n\rho^{\zeta}\Psi}{d \rho^2\sigma^2_{\epsilon}}, 1\right\}\right)  
\end{align*}
By using a union bound on all pairs of $\vec R$ and $\vec R'$ in $\bbS$, we have for sufficiently small $c_3$, there exists a $c_4$ such that 
\begin{align*}
\Pr\left[\exists \vec R, \vec R' \in \bbS: \mc(\vec R, \vec R') \geq \rho^{\zeta}\Psi\right] \leq \exp\left(-c_4 \rho^{2\lambda}d \max\left\{\frac{n\rho^{\zeta}\Psi}{d \rho^2\sigma^2_{\epsilon}}, 1\right\}\right)  \leq \exp(-c_4\rho^{2\lambda}d). 
\end{align*} 

\end{proof}

\subsubsection{Step 3. Building up unitary matrices}
We next construct a set of unitary matrices in $\reals^{d \times d}$ based on elements in $\bbS$.
We shall refer to the procedure to generate a unitary $U$ from an element $\vec R \in \bbS$ as $q(\vec R)$. 

\myparab{\underline{Procedure $q(\vec R)$:}} Let $U \in \reals^{d \times d}$ be the matrix $q(\vec R)$ aims to fill in.  
  Let $\mmv^{(1)}, \mmv^{(2)}, \dots, \mmv^{(\lt-1)} \in \reals^d$ be an arbitrary set of orthonormal vectors. 
$q(\vec R)$ partitions the matrix $U$ into three regions of columns and it fills different regions with different strategies. See also three regions illustrated by matrices in Fig.~\ref{fig:pack_U}. 

\mypara{Region 1: $U_{:, i}$ for $i < \lt$.} We set $U_{:, i} = \mmv^{(i)}$ when $i < \lt$. This means all the unitary matrices we construct
share the first $\lt - 1$ columns. 

\mypara{Region 2: $U_{:, i}$ for $i \in [\lt, \ut]$.} We next fill in non-zero entries of each $U_{:, i}$ by $\vec R_i$ for $i \in [\lt, \ut]$. We fill in each $U_{:, i}$ 
sequentially from left to right (from small $i$ to large $i$). We need to make sure that \emph{(i)} it is feasible to construct $U_{:, i}$ that is orthogonal to all $U_{:. j}$ ($j < i$)
by using only entries  specified by $R_i$. \emph{(ii)} there is a way to fill in $U_{:, i}$ so that not too many entries are excessively large. (ii) is needed because 
for any $\vec R$ and $\vec R'$, $\vec R_i$ and $\vec R'_i$ still share a small number of non-zero positions (whp $|\vec R_i \cap \vec R'_i| = O(\rho^{2\lambda}d)$, according to Lemma~\ref{lem:randomsubset}). When the mass in $|\vec R_i \cap \vec R'_i|$ is large, the distance between $U$ and $U'$ is harder to control. 
 
\mypara{Region 3: $U_{:, i}$ for $i > \ut$.} $q(\vec R)$ fills in the rest of the vectors arbitrarily so long as $U$ is unitary. Unlike the first $\lt - 1$ columns, these columns depend on $\vec R$ so each $U \in \bbU$ has a different set of ending column vectors.  
 
\myparab{Analysis for region 2.} Our analysis focuses on two asserted properties for Region 2 are true. We first show (i) is true and describe a procedure to make sure (ii) happens. 

For $j \leq i - 1$,  let $\mmw^{(j)} \in \reals^{\rho^{\lambda}d}$ be the projection of $U_{:, j}$ onto the coordinates specified by $\vec R_i$. See Fig.~\ref{fig:pack_U}(d) for an illustration. 
 Note that $\ut = \frac{\rho^{\lambda}d}{2}$, the dimension of the subspace spanned by $\mmw^{(1)}, \dots \mmw^{(i - 1)}$ is at most $\rho^{\lambda}d / 2$. 
Therefore, we can find a set of orthonormal vectors $\{\mmu^{(1)}, \dots \mmu^{(\kappa)}\} \subseteq \reals^{\rho^{\lambda}d}$ ($\kappa \geq \frac{\rho^{\lambda}d}{2}$) that are orthogonal to $\mmw^{(j)}$ ($j \leq i - 1$). 
To build a $U_{:, i}$ that's orthogonal to all $U_{:, j}$ ($j \leq i - 1$), we can first find a $\mmu \in \reals^{\rho^{\lambda}d}$ that is a linear combination of $\{\mmu^{(j)}\}_{j \leq \kappa}$, and then ``inflate'' $\mmu$ 
back to $\reals^{d}$, i.e., the $k$-th non-zero coordinate of $U_{:, i}$ is $\mmu_k$. One can see that 
\begin{align*}
\langle U_{:, j}, U_{:, i}\rangle = \langle \mmw^{(j)}, \mmu \rangle = 0
\end{align*}
for any $j < i$. 

We now make sure (ii) happens. We have the following Lemma.

\begin{lemma}\label{lem:tailbasis2} Let $\{\mmu^{(1)}, \dots, \mmu^{(\kappa)}\}$ ($\kappa \geq \rho^{\lambda}d/2$) be a collection of orthonormal vectors in $\reals^{\rho^{\lambda}d}$. 
Let $\eta$ be a small tunable parameter. 
There exists a set of coefficients 
$\beta_1, \dots, \beta_{\kappa}$ such that $\mmu = \sum_{i = 1}^{\kappa}\beta_i \mmu^{(i)}$ is a unit vector and there exist constant $c_5$, $c_6$, and $c_7$ such that 
\begin{equation}\label{eqn:tail}
\sum_{i \leq \rho^{\lambda}d} \mmu^2_i I\left(\mmu_i \geq \frac{c_5}{\sqrt{\rho^{\lambda+\eta}d} }\right) \leq c_7\xi, 
\end{equation}
where $\xi = \exp(-\frac{c_6}{\rho^{\eta}})$.
\end{lemma}

\begin{proof}[Proof of Lemma~\ref{lem:tailbasis2}] We use a probabilistic method to find $\mmu$. 
Let $z_i \sim N(0, 1/\sqrt{\rho^{\lambda}d})$ for $i \in [\rho^{\lambda}d]$. 
Let $S = \sqrt{\sum_{i \leq \rho^{\lambda}d}z^2_i}$. We shall set $\beta_i = z_i / S$. One can check that $\mmu = \sum_{i \in [\rho^{\lambda}d]}\beta_i \mmu^{(i)}$ 
is a unit vector. We then examine whether (\ref{eqn:tail}) is satisfied for these $\beta_i$'s we created. If not, we re-generate a new set of $z_i$'s and $\beta_i$'s. 
We repeat this process until (\ref{eqn:tail}) is satisfied. 

Because $\beta_i$'s are normalized, setting the standard deviation of $z_i$ is unnecessary. We 
nevertheless do so because $S$ will be approximately a constant, which helps us simplify the calculation. 

We claim that there exists a constant $c$ such that for any $\ell$,

\begin{equation}\label{eqn:singletail}
\E\left[\mmu^2_{\ell} I\left(\mmu_{\ell} \geq \frac{c_5}{\sqrt{\rho^{\lambda+\eta}d}}\right)\right] \leq \frac{c\xi}{\rho^{\lambda}d}.
\end{equation}

We first show that (\ref{eqn:singletail}) implies Lemma~\ref{lem:tailbasis2}. Then we will show (\ref{eqn:singletail}). 

By linearity of expectation, (\ref{eqn:singletail}) implies 
\begin{align*}
\E\left[\sum_{\ell\leq \rho^{\lambda}d}\mmu^2_{\ell} I\left(\mmu_{\ell} \geq \frac{c_5}{\sqrt{\rho^{\lambda+\eta}d}}\right)\right] \leq c\xi.
\end{align*}
Then we use a Markov inequality and obtain
\begin{align*}
\Pr\left[\left(\sum_{\ell\leq \rho^{\lambda}d}\mmu^2_{\ell} I\left(\mmu_{\ell} \geq \frac{c_5}{\sqrt{\rho^{\lambda+\eta}d}}\right) \right) \geq 2c\xi\right] \leq \frac 1 2
\end{align*}

So our probabilistic method described above is guaranteed to find a $\mmu$ that satisfies (\ref{eqn:tail})

We next move to showing (\ref{eqn:singletail}). Recall that 
\begin{align*}
\mmu_{\ell} = \frac 1 S (\sum_{i \leq \kappa}z_i \mmu^{(i)}_{\ell}). 
\end{align*}

We also let $Z_{\ell} = \sum_{i \leq \kappa}z_i \mmu^{(i)}_{\ell}$. We can see that $Z_{\ell}$ 
is a Gaussian random variable with a standard deviation $\sqrt{\frac{1}{\rho^{\lambda}d}\sum_{i \leq \kappa}(\mmu^{(i)}_{\ell})^2} \leq \sqrt{\frac 1{\rho^{\lambda}d}}$. The inequality
uses the fact that $\mmu^{(i)}$'s are orthonormal to each other and therefore $\sum_{i \leq \kappa}(\mmu^{(i)}_{\ell})^2 \leq 1$. 

On the other hand, one can see that 

\begin{align*}
\E[S] = \E[\sum_{i \leq \kappa}z^2_i] = \left(\frac{1}{\sqrt{\rho^{\lambda}d}}\right)^2 \cdot \kappa \geq \frac 1 2. 
\end{align*}

Therefore, by a standard Chernoff bound, $\Pr\left[S \leq \frac 1 4\right] \leq \exp(-\Theta(\kappa))$. 

Next, because $\{\mz_i\}_{i \leq \kappa}$ collectively form a spherical distribution, we have $\left\{\frac{z_i}{S}\right\}_{i \leq \kappa}$ is independent to $S$. Therefore, 

\begin{equation}\label{eqn:conditioning}
\E\left[\mmu^2_{\ell} I\left(\mmu_{\ell} \geq \frac{c_5}{\sqrt{\rho^{\lambda+\eta}d}}\right)\right] = \E\left[\mmu^2_{\ell} I\left(\mmu_{\ell} \geq \frac{c_5}{\sqrt{\rho^{\lambda+\eta}d}}\right) \mid S \geq \frac 1 4\right]
\end{equation}

Conditioned on $S \geq \frac 1 4$, we use the fact that $\mmu_{\ell} = Z_{\ell}/S$ to get 
 $\mmu_{\ell} \leq 4 Z_{\ell}$ and 
\begin{equation}
I\left(\mmu_{\ell} \geq \frac{c_5}{\sqrt{\rho^{\lambda+\eta}d}}\right) = I\left(Z_{\ell} \geq \frac{c_5}{\sqrt{\rho^{\lambda+\eta}d}}S\right) \leq I\left(Z_{\ell} \geq \frac{c_5}{4\sqrt{\rho^{\lambda+\eta}d}}\right). 
\end{equation}

Therefore, 
\begin{align*}
   &  (\ref{eqn:conditioning}) \\
    & \leq 16\E\left[Z^2_{\ell} I\left(Z_{\ell} \geq \frac{c_5}{4\sqrt{\rho^{\lambda+\eta}d}}\right) \mid S \geq \frac 1 4\right] \\
    & = \frac{16}{\Pr\left[S \geq \frac 1 4\right]}\left(\E\left[Z^2_{\ell} I\left(Z_{\ell} \geq \frac{c_5}{4\sqrt{\rho^{\lambda+\eta}d}}\right)\right] - \E\left[Z^2_{\ell} I\left(Z_{\ell} \geq \frac{c_5}{4\sqrt{\rho^{\lambda+\eta}d}}\right) \mid S < \frac 1 4\right]\Pr\left[S \leq \frac 1 4\right]\right) \\
    & \leq 16(1+\exp(-\Theta(\kappa)))\left(\E\left[Z^2_{\ell} I\left(Z_{\ell} \geq \frac{c_5}{4\sqrt{\rho^{\lambda+\eta}d}}\right)\right]\right) \\
    & \leq 16(1+\exp(-\Theta(\kappa)))\E\left[Z^2_{\ell} I\left(Z_{\ell} \geq \frac{2c_5}{\sqrt{\rho^{\lambda+\eta}d}}\right)\right] \\
& \leq\frac{16(1+\exp(-\Theta(\kappa)))}{\sqrt{2\pi}\sigma_{Z_{\ell}}} \int_{\frac{2c_5}{\sqrt{\rho^{\lambda+\eta}d}}}^{\infty}z^2 \exp\left(\frac{z^2}{\sigma^{2}_{Z_{\ell}}}\right)dz \\
& \leq  \frac{c}{\rho^{\lambda+\eta}d}\exp\left(-\Theta(\rho^{-\eta})\right) \quad {\mbox{(using $\sigma_{Z_{\ell}} \leq \sqrt{\frac{1}{\rho^{\lambda}d}}$)}}  \\
& \leq O(\frac{\xi}{\rho^{\lambda}d})
\end{align*}
\end{proof}

\subsubsection{Proof of Proposition~\ref{prop:lowerbound}}
Now we are ready to use Theorem~\ref{thm:block} to prove Proposition~\ref{prop:lowerbound}. 
Let $\bbS$ be the set constructed from Step 1 and $\bbU^{\normf} = \{U^{\normf}: U^{\normf} = q(\vec R), \vec R \in \bbS\}$. Let $\bbN^{\normf} = \{U^{\normf}L^{\normf}: U^{\normf} \in \bbU^{\normf}\}$.
Let also $\mP_{N^{\normf}}$ be the distribution of $(\my, \mz^{\normf})$ generated from the model $\my = N^{\normf} \mz^{\normf} + \epsilon$. Let $\mP_{N^{\normf}\mid \mz^{\normf}}$ be the distribution of 
$\my$ from the normalized model when $\mz^{\normf}$ is given. Let $\mP_{N^{\normf}, n}$ be the product distribution when we observe $n$ samples from the model $\my = N^{\normf} \mz^{\normf} + \epsilon$. 
Let $f_{N^{\normf}}(\my, \mz^{\normf})$ be the pdf of for $\mP_{N^{\normf}}$, $f_{N^{\normf}}(\my \mid \mz^{\normf})$ be the pdf of $\my$ given $\mz^{\normf}$, and $f(\mz^{\normf})$ be the pdf of $\mz^{\normf}$. 

We need to show that  for any $N^{\normf}, \tilde N^{\normf} \in \bbN$ \emph{(i)} $\|N^{\normf} - \tilde N^{\normf}\|_F$ is large, and \emph{(ii)} the KL-divergence between
 $\mP_{N^{\normf}, n}$ and $\mP_{N^{\normf}, n}$ is bounded.

\begin{lemma}\label{lem:lowernorm} Let  $\bbS$, $\bbU^{\normf}$, and $\bbN^{\normf}$ be generated by the parameters $\lambda, \eta$, and $\zeta$ (see Steps 1 to 3). 
For any $N^{\normf}$ and $\tilde N^{\normf}$ in $\bbN$, we have 
\begin{equation}
\|N^{\normf} - \tilde N^{\normf}\|^2_F \geq \sum_{i \in [\lt, \ut]}\sigma^{\normf}_i (2 - c_8 \rho^{\lambda - \eta} - c_9\rho^{\zeta})
\end{equation}
for some constant $c_8$ and $c_9$. 
\end{lemma}

\begin{proof}
Let $U^{\normf}, \tilde U^{\normf} \in \bbU^{\normf}$ such that $N^{\normf} = U^{\normf}L^{\normf}$ and $\tilde N^{\normf} = \tilde U^{\normf}L^{\normf}$; let $\vec R, \vec R' \in \bbS$ be that $U^{\normf} = q(\vec R)$ and $\tilde U^{\normf} = q(\vec R')$. 
Also, recall that we let $\Psi = \sum_{i \in [\lt, \ut]}(\sigma^{\normf}_i)^2$. 

We have 
$$\|N^{\normf} - \tilde N^{\normf}\|^2_F = \sum_{i \in [d]} \|(U^{\normf}_{:, i} - \tilde U^{\normf}_{:, i})\sigma^{\normf}_i\|^2_2 \geq \sum_{i \in [\lt, \ut]} \|(U^{\normf}_{:, i} - \tilde U^{\normf}_{:, i})\sigma^{\normf}_i\|^2_2$$

Let also  $H = \{ \vec R_i = \vec R'_i, i \in [\lt, \ut]\}$, i.e., the set of coordinates that $\vec R$ and $\vec R'$ agree. Whp, we have 
\begin{equation}
\Psi = \sum_{i \in H}(\sigma^{\normf}_i)^2 + \sum_{{\scriptsize \substack{i \in [\lt, \ut] \\ i \notin H}}}(\sigma^{\normf}_i)^2 = \mc(\vec R, \vec R') + \sum_{{\scriptsize \substack{i \in [\lt, \ut] \\ i \notin H}}}(\sigma^{\normf}_i)^2 \leq \rho^{\zeta}\Psi + \sum_{{\scriptsize \substack{i \in [\lt, \ut] \\ i \notin H}}}(\sigma^{\normf}_i)^2
\end{equation}
The last inequality holds because of Lemma~\ref{lem:hamming}. Therefore, we have $ \sum_{{\scriptsize \substack{i \in [\lt, \ut] \\ i \notin H}}}(\sigma^{\normf}_i)^2 \geq (1-\rho^{\zeta})\Psi$. 

Now we have
$$
\sum_{i \in [\lt, \ut]}\|(U^{\normf}_{:, i} - \tilde U^{\normf}_{:, i})\sigma^{\normf}_i\|^2_2  \geq \sum_{{\scriptsize \substack{i \in [\lt, \ut] \\ i \notin H}}} \|(U^{\normf}_{:, i} - \tilde U^{\normf}_{:, i}\|^2_2 (\sigma^{\normf}_i)^2. 
$$

Next we bound $\| U^{\normf}_{:, i} - \tilde U^{\normf}_{:, i}\|^2_2$ when $\vec R_i \neq \vec R'_i$. We have
\begin{equation}
\|U^{\normf}_{:, i} - \tilde U^{\normf}_{:, i}\|^2_2 \geq 2 - \sum_{j \in \vec R_i \cap \vec R'_i} \left((U^{\normf}_{j, i})^2 + (\tilde U^{\normf}_{j, i})^2\right). 
\end{equation}
By Lemma~\ref{lem:randomsubset}, whp $|\vec R_i \cap \vec R'_i | \leq  c_2\rho^{2\lambda}d$. Next, we give a bound for $\sum_{j \in \vec R_i \cap \vec R'_i}(U^{\normf}_{j, i})^2$. 
The bound for  $\sum_{j \in \vec R_i \cap \vec R'_i}(\tilde U^{\normf}_{j, i})^2$ can be derived in a similar manner. 

For each $j \in |\vec R_i \cap \vec R'_i|$, we check whether $U^{\normf}_{j, i} \geq \frac{c_5}{\sqrt{\rho^{\lambda+\eta}d}}$: 
\begin{align*}
 \sum_{j \in \vec R_i \cap \vec R'_i} (U^{\normf}_{j, i})^2 & =  \sum_{j \in \vec R_i \cap \vec R'_i}\left[ (U^{\normf}_{j, i})^2 I(U^{\normf}_{j, i}\leq \frac{c_5}{\sqrt{\rho^{\lambda+\eta}d}}) + (U^{\normf}_{j, i})^2 I(U^{\normf}_{j, i}>  \frac{c_5}{\sqrt{\rho^{\lambda+\eta}d}})\right] \\
& \leq  O\left(\frac{\rho^{2\lambda}d}{\rho^{\lambda+\eta}d} + \exp(-\Theta(1/\rho^{\eta}))\right)  = O(\rho^{\lambda - \eta}). 
\end{align*}

Therefore, 
\begin{align*}
\sum_{i \in [\lt, \ut]}\|(U^{\normf}_{:, i} - \tilde U^{\normf}_{:, i})(\sigma^{\normf}_i)^2\|^2_2  & \geq   \sum_{i \in [\lt, \ut]} (\sigma^{\normf}_i)^2 (2- O(\rho^{\lambda - \eta})(1-\rho^{\zeta}))  \geq  \sum_{i \in [\lt, \ut]} (\sigma^{\normf}_i)^2 (2- c_8\rho^{\lambda - \eta} - c_9\rho^{\zeta}).  
\end{align*}
\end{proof}

We need two additional building blocks.

\begin{lemma}\label{lem:tailN} Consider the regression problem $\my = M\mx + \epsilon$, where $\|M\|\leq \Upsilon = O(1)$ and the eigenvalues $\sigma_i(\E[\mx \mx^{\transpose})]$ of the features follow a power law distribution with exponent $\omega$. Consider the problem $\my = N \mz + \epsilon$ in orthogonalized form. Let $\sigma^N_i$ be the $i$-th singular value of $N$. Let $t$ be an arbitrary value in $[0, \min \{d_1, d_2\}]$. There exists a constant $c_{10}$ such that 
$$\sum_{i \geq t}\sigma^2_{i}(N) \leq \frac{c_{10}}{t^{\omega - 1}}.$$
\end{lemma}

\begin{proof}[Proof of Lemma~\ref{lem:tailN}] 
Without loss of generality,  assume that $d_1 \geq d_2$. We first split the columns of $N$ into two parts, 
$N = [N_+, N_-]$, in which $N_+ \in \reals^{d_2 \times t}$ consists of the first $t$ columns of $N$ and $N_- \in \reals^{d_2 \times (d_1 - t)}$ consists of the remaining columns. Let $\mzero$ be a zero matrix in  $\reals^{d_2 \times (d_1-t)}$. $[N_+, \mzero]$ is a matrix of rank at most $t$.

Let us split other matrices in a similar manner. 
\begin{itemize}
\item  $M = [M_+, M_-]$, where $M_+ \in \reals^{d_2 \times t}$ and $M_- \in \reals^{d_2 \times (d_1 - t)}$,
\item $V^* = [V^*_+, V^*_-]$, where $V^*_+ \in \reals^{d_1 \times t}$, 
 and $V^*_- \in \reals^{d_1 \times (d_1 - t)}$, and
\item $\Lambda^* = [\Lambda^*_+, \Lambda^*_-]$, where $\Lambda^*_+ \in \reals^{d_1 \times t}$ and $\Lambda^*_- \in \reals^{d_1 \times (d_1 - t)}$. 
\end{itemize}
We have 
\begin{align*}
\sum_{i \geq t}(\sigma^N_{i})^2 & =  \| N - \mP_{t}(N)\|^2_F \\
& \leq  \| N - [N_+, \mzero]\|^2_F \quad \mbox{ ($\mP_{t}(N)$ gives an optimal rank-$t$ approximation of $N$).} \\
& =  \|N_-\|^2_F \leq  \|M_- V^*_-\|^2 \|(\Lambda^*_{-})^{\frac 1 2}\|^2_F  \leq  \frac{c_{10}}{t^{\omega - 1}}
\end{align*}
\end{proof}

We next move to our second building block. 

\begin{fact} \label{fact:kd}
\begin{equation}
\Kld(\mP_{N^{\normf}, n}, \mP_{\tilde N^{\normf}, n}) = \frac{n \|N^{\normf} - \tilde N^{\normf}\|^2_F}{2\sigma^2_{\epsilon}}.
\end{equation} 
\end{fact}

\begin{proof}[Proof of Fact~\ref{fact:kd}]
 \begin{eqnarray*}
 & & \Kld(\mP_{N^{\normf}, n}, \mP_{\tilde N^{\normf}, n}) \\
 & = & n \Kld(\mP_{N^{\normf}}, \mP_{\tilde N^{\normf}}) \\
 & = & n \E_{\my = N^{\normf}\mz^{\normf} + \epsilon} \left[\log\left(\frac{ f_{N^{\normf}}(\my,\mz^{\normf})}{ f_{\tilde N^{\normf}}(\my, \mz^{\normf})}\right)\right] \\
 & = & n \E_{\mz^{\normf}} \left[\E_{\my = N^{\normf}\mz^{\normf} + \epsilon}\left[\log\left(\frac{ f_{N^{\normf}}(\my, \mz^{\normf})}{ f_{\tilde N^{\normf}}(\my, \mz^{\normf}) }\right)\mid\mz^{\normf}\right]\right] \\
 & = &  n \E_{\mz^{\normf}} \left[\E_{\my = N^{\normf}\mz^{\normf} + \epsilon}\left[\log\left(\frac{ f_{N^{\normf}}(\my \mid \mz^{\normf})f(\mz^{\normf})}{ f_{\tilde N^{\normf}}(\my \mid \mz^{\normf})f(\mz^{\normf}) }\right)\mid\mz^{\normf}\right]\right] \\
 & = & n \E_{\mz}\E_{\my = N^{\normf}\mz^{\normf} + \epsilon}\left[\left[\log\left(\frac{f_{N^{\normf}}(\my\mid \mz^{\normf})}{f_{\tilde N^{\normf}}(\my \mid \mz^{\normf})}\right) \mid \mz^{\normf}\right]\right] \\
 & = & n \E_{\mz}\left[\Kld(\mP_{N^{\normf}\mid \mz^{\normf}}, \mP_{\tilde N^{\normf} \mid \mz^{\normf}})\right]  =  \frac{n \|N^{\normf} -\tilde N^{\normf}\|^2_F}{2\sigma^2_{\epsilon}}
 \end{eqnarray*}
 \end{proof}

We now complete the proof for Proposition~\ref{prop:lowerbound}. First, define $\psi \triangleq \frac{\sum_{i \geq \ut + 1}(\sigma^{\normf}_i)^2}{\sum_{i \in [\lt, \ut]}(\sigma^{\normf}_i)^2}$. For any 
$N^{\normf}$ and $\tilde N^{\normf}$ in $\bbN$, we have 
$$\|N^{\normf} - \tilde N^{\normf}\|^2_F = \sum_{i \geq \lt} (\sigma^{\normf}_i)^2 = (1+\psi)\Psi,$$
where recall that $\Psi = \sum_{i \in [\lt, \ut]}(\sigma^{\normf}_i)^2$. Using Lemma~\ref{lem:lowernorm}, we have $\Kld(\mP_{N^{\normf}, n}, \mP_{\tilde N^{\normf}, n}) = n (1+\psi)\Psi$. 
 
Next, we find a smallest $\alpha$ such that  
\begin{equation}\label{eqn:findalpha}
\max_{N^{\normf}, \tilde N^{\normf}} \Kld(\mP_{N^{\normf}, n}, \mP_{\tilde N^{\normf}, n}) \leq \alpha \log |\bbN|. 
\end{equation}
 
By Lemma~\ref{lem:hamming}, we have $|\bbN| = \exp(c_3 \frac{n\rho^{2\lambda+\zeta-2}}{\sigma^2_{\epsilon}}\Psi)$. (\ref{eqn:findalpha}) is equivalent to 
requiring 
$$\frac{n (1+\psi)\Psi}{2\sigma^2_{\epsilon}} \leq \frac{\alpha c_3 n \rho^{2\lambda + \zeta - 2} \Psi}{\sigma^2_{\epsilon}}.$$
We may thus set $\alpha = O(\rho^{2 - 2\lambda + \zeta }(1+\psi))$. Now we may invoke Theorem~\ref{thm:block} and get
\begin{align*}
\mr(\mx, M, n, \sigma_{\epsilon}) & \geq \Psi\left(1-\frac{1}{\sqrt{|\bbN|}}\right)\underbrace{(2-c_8\rho^{\lambda - \eta} - c_9 \rho^{\zeta})}_{\mbox{Lemma~\ref{lem:lowernorm}}}\left(1-O(\rho^{2 - 2\lambda - \zeta})(1+\psi)\right) \\
& \geq \Psi (1 - O(\rho^{\lambda - \eta} + \rho^{\zeta} + \rho^{2 - 2\lambda - \zeta})) - \underbrace{\Psi \psi}_{= \sum_{i > \ut}(\sigma^{\normf}_i)^2} \rho^{2 - 2\lambda - \zeta} \\
& \geq \Psi(1- O(\rho^{\lambda - \eta} + \rho^{\zeta} + \rho^{2 - 2\lambda - \zeta})) - O\left(\frac{1}{(\rho^{\lambda}d)^{\omega - 1}}\right)\rho^{2 - 2\lambda - \zeta} \\
&  \quad \quad \mbox{ (Use Lemma~\ref{lem:tailN} to bound $\sum_{i > \ut}(\sigma^{\normf}_i)^2$ )} 
\end{align*}

We shall set $\xi = \eta$ be a small constant (say $0.001$), $\lambda = \frac 1 2 + \xi$, and $\zeta = \frac 1 2$. This gives us 
$$ \mr(\mx, M, n, \sigma_{\epsilon}) \geq \Psi(1-O(\rho^{\frac 1 2} + \rho^{\frac 1 2 - 2\xi})) - \frac{\rho^{\frac 1 2 - 2\xi - (\omega - 1)(\frac 1 2 + \xi)}}{d^{\omega - 1}}$$ 
Together with the fact that $\sum_{i > \ut}(\sigma^{\normf}_i)^2 = O\left(\frac{1}{(\rho^{\lambda} \omega)^{\omega - 1}}\right)$, we complete the proof of Proposition~\ref{prop:lowerbound}.



\fi

\ifnum\full=0

\fi
\section{Related work and comparison}\label{sec:related}
\ifnum\full=0

\fi

In this section, we compare our results to other regression algorithms that
make low rank constraints on $M$.  Most existing MSE results are parametrized
by the rank or spectral properties of $M$, e.g.
\ifnum\full=0
~\cite{negahban2011estimation} \else~\cite{negahban2011estimation}\fi  defined a generalized notion of rank 
\ifnum\full=0
$
    \bbB_q(R^A_q) \in \bigl\{A \in \reals^{d_2 \times d_1}: \sum_{i = 1}^{d_2}|\sigma^A_i|^q \leq R_q\bigr\},
$
where $q \in [0, 1], A \in \{N, M\}$,
\else
\begin{equation}
    \bbB_q(R^A_q) \in \bigl\{A \in \reals^{d_2 \times d_1}: \sum_{i = 1}^{\min\{d_1, d_2\}}|\sigma^A_i|^q \leq R_q\bigr\}, \quad q \in [0, 1], A \in \{N, M\}, 
\end{equation}
\fi
i.e. $R^N_q$ characterizes the generalized rank of $N$ whereas $R^M_q$ characterizes that of $M$. When $q = 0$, $R^N_q = R^M_q$ is the rank of the $N$ because $\Rank(N) = \Rank(M)$ in our setting. 
In their setting, the MSE is parametrized by $R^M$ and is shown to be
$O\Bigl(R^M_q\Bigl(\frac{\sigma^2_{\epsilon}\lambda^*_1(d_1+d_2)}{(\lambda^*_{\min})^2n}\Bigr)^{1-q/2}\Bigr)$.
In the special case when $q = 0$, this reduces to $O\left(\frac{\sigma^2_{\epsilon}\lambda^*_1\Rank(M)(d_1+d_2)}{(\lambda^*_{\min})^2\cdot n}\right)$.
On the other hand, the MSE in our case is bounded by (cf. Thm.~\ref{thm:main_upper}). We have 
$ \E[\|\hat \my - \my\|^2_2] = O\bigl(R^N_q (\frac{\sigma^2_{\epsilon}d_2}{n})^{1-q/2} + n^{-c_0}\bigr).
$
%
When $q = 0$, this becomes $O\bigl(\frac{\sigma^2_{\epsilon}\Rank(M)d_2}{n} + n^{-c_0}\bigr)$.

The improvement here is twofold. First, our bound is directly characterized by $N$ in orthogonalized form, whereas result of
\ifnum\full=0~\cite{negahban2011estimation} \else~\cite{negahban2011estimation}\fi 
needs to examine the interaction between $M$ and $C^*$, so their MSE depends on
both $R^M_q$ and $\lambda^*_{\min}$. Second, our bound no longer depends on
$d_1$ and pays only an additive factor $n^{-c_0}$, thus, when $n < d_1$, our
result is significantly better. Other works have different parameters in the
upper bounds, but all of these existing results require that $n > d_1$
 to obtain non-trivial upper
bounds~\cite{koltchinskii2011nuclear,bunea2011optimal,chen2013reduced,koltchinskii2011nuclear}. Unlike these prior work, we require a stochastic assumption on $\mX$ (the rows are i.i.d.) to ensure that the model is identifiable when $n < d_1$, e.g. there could be two sets of disjoint features that fit the training data equally well. 
Our algorithm produces an adaptive model whose complexity is controlled by
$k_1$ and $k_2$, which are adjusted dynamically depending on the sample size
and noise level. \cite{bunea2011optimal} and \cite{chen2013reduced} also point
out the need for adaptivity; however they still require $n > d_1$ and make
some strong assumptions. For instance, 
\ifnum\full=0
\cite{bunea2011optimal}
\else\cite{bunea2011optimal}\fi assumes that there is a gap between $\sigma_i(\mX M^{\transpose})$ and $\sigma_{i + 1}(\mX M^{\transpose})$ for some $i$. In comparison, our sufficient condition, the decay of $\lambda^*_i$, is more natural.
Our work is not directly comparable to standard variable selection techniques such as LASSO~\cite{tibshirani1996regression} because they handle univariate $\my$. 
Column selection algorithms~\cite{dan2015low}  generalize variable selection methods for vector responses, but they cannot address the identifiability concern.



\ifnum\full=1

\subsection{Missing proof in comparison}

This section proves the following corollary. 
\begin{corollary}\label{cor:mserank} Use the notation appeared in Theorem~\ref{thm:main_upper}. Let the ground-truth matrix $N \in \bbB_q(R^N_q)$ for $q \in [0, 1]$. We have whp
\begin{equation}
    \E[\|\hat \my - \my\|^2_2] = O\left(R^N_q\left(\frac{\sigma^2_{\epsilon}d_2}{n}\right)^{1-q/2} + n^{-c_0}\right).
\end{equation}
\end{corollary}

\begin{proof} We first prove the case $q = 0$ as a warmup. 
Observe that 
$$\|N\|^2_F - \sum_{i \leq \ell^*}(\sigma^N_i)^2 = \sum_{\ell^* < i \leq r}(\sigma^N_i)^2 \leq \frac{\theta^2 \sigma^2_{\epsilon} d_2 (r- \ell)}{n} .$$
The last inequality uses $\sigma^N_i \leq \theta \sigma_{\epsilon}\sqrt{\frac{d_2}{n}}$ for $i > \ell^*$. Therefore, we have
$$\E[\|\hat \my - \my \|^2)2] \leq O\left(\frac{(r - \ell^*)\theta^2 \sigma^2_{\epsilon}d_2}{n} + \frac{\ell^* d_2 \theta^2 \sigma^2_{\epsilon}}{n} + n^{-c_0}\right)= O\left(\frac{r\theta^2 \sigma^2_{\epsilon}d_2}{n} + n^{-c_0}\right).$$

Next, we prove the general case $q \in (0, 1]$. We can again use an optimization view to give an upper bound of the MSE. We view $\sum_{i \leq d_2}(\sigma^N_i)^q \leq R_q$ as a constraint. We aim to maximize the uncaptured signals, i.e., solve the following optimization problem

\begin{align*}
    \mbox{maximize:} \quad \quad   & \sum_{i > \ell^*} (\sigma^N_i)^2 \\
    \mbox{subject to:} \quad \quad  & \sum_{i > \ell}(\sigma^N_i)^q \leq R_-, \quad \mbox{ where $R_- = R_q - \sum_{i \leq \ell^*}(\sigma^N_i)^q$} \\
    & \sigma^N_i \leq \frac{\theta^2 \sigma^2_{\epsilon}d_2}{n}, \quad \mbox{ for $i \geq \ell^*$.}
\end{align*}

The optimal solution is achieved when $(\sigma^N_i)^2 = \frac{\theta^2 \sigma^2_{\epsilon}d_2}{n}$ for $\ell^* < i \leq \ell^* + k$, where $k = \frac{R_-}{\left(\frac{\theta^2 \sigma^2_{\epsilon}d_2}{n}\right)^{\frac q 2}}$, and $\sigma_i^N = 0$ for $i > \ell^* + k$. We have 

\begin{align}
    & \E[\|\hat \my - \my\|^2_2] \nonumber \\
    & \leq \sum_{i > \ell^*}(\sigma^N_i)^2 + O\left(\frac{\ell^* d_2 \theta^2 \sigma^2_{\epsilon}}{n} + n^{-c_0}\right) \nonumber \\
    & = \left(R_q - \sum_{i \leq \ell^*}(\sigma^N_i)^q\right)\left(\frac{\theta^2 \sigma^2_{\epsilon}d_2}{n}\right)^{1-q/2} + O\left(\frac{\ell^* d_2 \theta^2 \sigma^2_{\epsilon}}{n} + n^{-c_0}\right) \label{eqn:genrank} 
\end{align}
We can also see that 
$$\sum_{i \leq \ell^*}(\sigma^N_i)^q \left(\frac{\theta^2 \sigma^2_{\epsilon}d_2}{n}\right)^{1-q/2} \geq \sum_{i \leq \ell^*}\left(\frac{\theta^2 \sigma^2_{\epsilon}d_2}{n}\right)^{q/2}\left(\frac{\theta^2 \sigma^2_{\epsilon}d_2}{n}\right)^{1-q/2} = \ell^*\left(\frac{\theta^2\sigma^2_{\epsilon}d_2}{n}\right).$$
Therefore,
$$(\ref{eqn:genrank})\leq  O\left(R_q\left(\frac{\theta^2 \sigma^2_{\epsilon}d_2}{n}\right)^{1-q/2} + n^{-c_0}\right).$$
\end{proof}

\fi

\section{Experiments}\label{sec:exp}
\ifnum\full=1
\begin{table*}[!htbp]
\centering 
\resizebox{0.8\textwidth}{!}
{\begin{minipage}{\textwidth}
\centering 
\begin{tabular}{l|l|l|l|l|l|l|l} \hline
Model                 & MSE$_{out}$         & MSE$_{in}$         & $\text{MSE}_{out - in}$ & $R^2_{out}$(bps)        & $R^2_{in}$ (bps)           & Sharpe             & $t$-statistic     \\ \hline
ARRR, N= 983       & \textbf{0.9935} & 1.0140 & \textbf{-0.0205}  & \textbf{46.3761} & 158.2564  & \textbf{2.4350} & \textbf{8.3268}  \\
Lasso              & 1.1158 & 0.3953 & 0.7205 & 6.6049  & 7147.0116 & 2.1462 & 0.0601  \\ 
Ridge              & 1.2158 & 0.1667 & 1.0491 & 9.8596  & 8511.9076 & 0.6603 & -0.0497 \\ 
Reduced ridge & 1.0900 & 0.8687 & 0.2213 & 13.0321 & 1555.5136 & 0.3065 & -0.3275 \\ 
RRR &  1.2200 & 0.5867 & 0.6332 & 7.0830 & 4121.2548 & 0.3647 & -0.6626\\ 
Nuclear norm &  1.2995 & 0.12078 & 1.1787 & 4.7297 & 8789.0625 & 0.6710 & 0.2340\\
PCR &  1.0259 & 0.8456 & 0.1802 & 1.1278 & 1544.7258 & 1.8070 & 0.3947\\

\hline

ARRR, N= 2838      & \textbf{1.0056} & 0.9050 & \textbf{0.1006} & \textbf{18.5761} & 689.0625  & \textbf{1.6239} & \textbf{15.4134} \\ 
Lasso              & 1.0625 & 0.5286 & 0.5339 & 1.1236  & 6029.5225 & 0.5954 & 0.0179  \\ 
Ridge              & 1.0289 & 0.6741 & 0.3548 & 0.2116  & 5342.1481 & 0.5739 & 0.0670  \\ 
Reduced ridge & 1.9722 & 0.7373 & 1.2349 & 1.0816  & 2416.7056 & 1.5482 & 0.0619  \\ 
RRR & 1.0873 & 0.61376 & 0.4735 & 4.5795 & 3844.124 & -0.477 & 0.6399\\
Nuclear norm & 1.1086 & 0.15346 & 0.9551 & 2.2097 & 8461.2402 & -0.3698 & -0.8986 \\
PCR & 1.0263  & 0.5336 & 0.4927 & 5.233 &  4653.9684& 1.2799 & 0.6990 \\
\hline
\end{tabular}
\captionsetup{font=normal}
\captionsetup{justification=centering}
\caption{{Summary of results for equity return forecasts. $R^2$ are measured by basis points (bps). $1\mathrm{bps} = 10^{-4}$.  Bold font denotes the best \emph{out-of-sample} results and smallest gap.}}
\label{table:exp_summary}
\end{minipage}}
\end{table*}
\begin{table*}[h!]
\centering
\resizebox{0.8\textwidth}{!}{\begin{minipage}{\textwidth}
\centering
\begin{tabular}{l|l|l|l|l|l}
\hline
Model    & MSE$_{in}$ & MSE$_{out}$ & MSE$_{out-in}$ & Corr$_{in}$ & Corr$_{out}$        \\ \hline
ARRR      & 5.0104 $\pm$ {0.38}  &  \textbf{9.4276} $\pm$ 2.31 &  \textbf{4.4172}  & 0.7425 $\pm$ 0.07 & \textbf{0.6730} $\pm$ 0.13 \\
Lasso     & 2.3755  $\pm$ 1.95 & 14.8279 $\pm$ 4.81 & 12.4524 &  0.9171  $\pm$ 0.09 & 0.4754  $\pm$  0.15\\
Ridge     &1.3974 $\pm$ {0.53}  & 13.6244 $\pm$ 4.39 & 12.2270 & 0.9555  $\pm$ 0.04 & 0.4742 $\pm$ 0.17 \\
Reduced ridge & 4.5260 $\pm$ {1.93}  & 12.2339 $\pm$ 2.70  & 7.7079  & 0.7905 $\pm$ 0.09 & 0.4972 $\pm$ 0.18 \\
RRR       & 4.3456 $\pm$ {0.47} & 13.0768 $\pm$ 2.63 & 8.7313  & 0.7725 $\pm$ 0.12 & 0.3820 $\pm$ 0.22\\
Nuclear norm  & 4.9190 $\pm$ {2.04} & 13.0532 $\pm$ 4.38 & 8.6677 & 0.7872  $\pm$ 0.10 & 0.4869 $\pm$ 0.16\\
PCR   & 6.4037 $\pm$ {1.99} &13.0847$\pm$ 4.19 & 8.8892  & 0.7199 $\pm$ 0.05 & 0.4861 $\pm$ 0.15 \\
 \hline
\end{tabular}
\captionsetup{font=normal}
\captionsetup{justification=centering}
\caption{Average results for Twitter dataset from 10 random samples. Bold font denotes the best \emph{out-of-sample} results and smallest gap}
\label{table:main_tweets}
\end{minipage}}
\end{table*}
\fi
\ifnum\full=0
\begin{table}
\centering
\captionsetup{font=small}
\captionsetup{width=1.0\linewidth}
\captionsetup{justification=centering}
\caption{{Summary of results for equity return forecasts (left) and average results for Twitter (right) from 10 random samples. $R^2$ are measured by basis points (bps). $1\mathrm{bps} = 10^{-4}$.  Bold font denotes the best \emph{out-of-sample} results and smallest gap. $out-in$ denotes MSE$_{out-in}$}}. 
\label{table:exp_summary}
\begin{adjustbox}{width=0.95\columnwidth,center}
\begin{tabular}{c|r|r|r|r|r|r|r|r} 
\hline
\multicolumn{1}{c}{}   & \multicolumn{5}{c|}{Equity return}                                                                                                                                & \multicolumn{3}{c}{Twitter dataset}                                                                     \\ 
\hline
Model                  & \multicolumn{1}{c|}{$R^2_{out}$ } & \multicolumn{1}{c|}{Sharp} & \multicolumn{1}{c|}{t-stat} & \multicolumn{1}{c|}{MSE$_{out}$ } & \multicolumn{1}{c|}{$out-in$ } & \multicolumn{1}{c|}{corr$_{out}$ } & \multicolumn{1}{c|}{MSE$_{out}$ } & \multicolumn{1}{c}{$out-in$ }  \\ 
\hline
$\proc{Adaptive-RRR}$  & \textbf{18.576}                   & \textbf{1.623}             & \textbf{15.413}             & \textbf{1.005}                    & \textbf{0.1006}                & \textbf{0.67} $\pm$ 0.13           & \textbf{9.42} $\pm$ 2.31          & \textbf{4.417}                 \\
Lasso                  & 1.124                             & 0.595                      & 0.018                       & 1.063                             & 0.534                          & 0.47 $\pm$ 0.15    & 14.82 $\pm$ 4.81                  & 12.452                         \\
Ridge                  & 0.212                             & 0.574                      & 0.067                       & 1.029                             & 0.355                          & 0.47 $\pm$ 0.17                    & 13.62 $\pm$ 4.39                  & 12.2\textbf{} 27               \\
Reduced ridge          & 1.082                             & 1.548                      & 0.062                       & 1.972                             & 1.235                          & 0.49 $\pm$ 0.18                    & 12.23 $\pm$ 2.70                  & 7.708                          \\
RRR                    & 4.580                             & -0.477                     & 0.640                       & 1.087                             & 0.474                          & 0.38 $\pm$ 0.22                    & 13.07 $\pm$ 2.63                  & 8.731                          \\
Nuclear norm           & 2.210                             & -0.370                     & -0.899                      & 1.109                             & 0.955                          & 0.48 $\pm$ 0.16                    & 13.05 $\pm$ 4.38                  & 8.668                          \\
PCR                    & 5.233                             & 1.280                      & 0.699                       & 1.026                             & 0.493                          & 0.48 $\pm$ 0.15                    & 13.08 $\pm$ 4.19                  & 8.889                          \\
\hline
\end{tabular}
\end{adjustbox}
\vspace{-5mm}
\end{table}

\fi
We apply our algorithm on an equity market and a social network dataset
to predict equity returns and user popularity respectively.
Our baselines include ridge regression (``Ridge''), reduced rank ridge regression \cite{mukherjee2011reduced} (``Reduced ridge''), LASSO (``Lasso''), nuclear norm regularized regression (``Nuclear norm''), reduced rank regression \cite{velu2013multivariate} (``RRR''), and principal component regression~\cite{agarwal2019robustness} (``PCR'').

\myparab{Predicting equity returns.}
We use a stock market dataset from an emerging market that consists of approximately 3600 stocks between 2011 and 2018. We focus on predicting the \emph{next 5-day returns}. For each asset in the universe, we compute its past 1-day, past 5-day, and past 10-day returns as features.  We use a standard approach to translate forecasts into positions~\cite{ang2014asset,zhou2014active}. We examine two  universes in this market: (i) \emph{Universe 1} is equivalent to S$\&$P 500 and consists of 983 stocks, and (ii) \emph{Full universe} consists of all stocks except for illiquid ones. 

\ifnum\full=0
\mypara{Results.} Table~\ref{table:exp_summary} (left)
reports the forecasting power and portfolio return for \emph{out-of-sample} periods in \emph{Full universe} (see our full version for \emph{Universe 1}).  We observe that \emph{(i)} The data has a low signal-to-noise ratio. The out-of-sample $R^2$ values of all the methods are close to 0. \emph{(ii)} $\proc{Adaptive-RRR}$ has the highest forecasting power. 
\emph{(iii)} $\proc{Adaptive-RRR}$ has the smallest in-sample and out-of-sample gap (see column ${out - in}$), suggesting that our model is better at avoiding spurious signals. 

\myparab{Predicting user popularity in social networks.}
We collected tweet data on political topics from Oct. 2016 to Dec. 2017.
Our goal is to predict a user's \emph{next 1-day} popularity, which is defined as the sum of retweets, quotes, and replies received by the user. There are a total of 19 million distinct users, and due to the huge size, we extract the subset of 2000 users with the most interactions for evaluation.
For each user in the 2000-user set, we use its past 5 days' popularity as features. We further randomly sample 200 users and make predictions for them, i.e., setting $d_2 = 200$ to make $d_2$ of the same magnitude as $n$. 

\ifnum\full=0
\fi
\mypara{Results.}
We randomly sample users for 10 times and report the average MSE and correlation (with standard deviations) for both \emph{in-sample} and \emph{out-of-sample} data (see  full version for more results). In Table~\ref{table:exp_summary} (right) we can see results consistent with the equity returns experiment: \emph{(i)}
$\proc{Adaptive-RRR}$ yields the best performance in out-of-sample MSE and correlation. \emph{(ii)}
$\proc{Adaptive-RRR}$ achieves the best generalization error by having a much smaller gap between training and test metrics.
\fi

\ifnum\full=1

\myparab{Predicting equity returns.}
We use a stock market dataset from an emerging market that consists of approximately 3600 stocks between 2011 and 2018. We focus on predicting the \emph{next 5-day returns}. For each asset in the universe, we compute its past 1-day, past 5-day and past 10-day returns as features.  We use a standard approach to translate forecasts into positions~\cite{ang2014asset,zhou2014active}. We examine two  universes in this market: (i) \emph{Universe 1} is equivalent to S$\&$P 500 and consists of 983 stocks, and (ii) \emph{Full universe} consists of all stocks except for illiquid ones. 

\mypara{Results.} Table~\ref{table:exp_summary} reports the forecasting power and portfolio return for \emph{out-of-sample} periods in two  universes.  We observe that \emph{(i)} The data has a low signal-to-noise ratio. The out-of-sample $R^2$ values of all the methods are close to 0. \emph{(ii)} $\proc{Adaptive-RRR}$ has the highest forecasting power. 
\emph{(iii)} $\proc{Adaptive-RRR}$ has the smallest in-sample and out-of-sample gap (see column MSE$_{out - in}$), suggesting that our model is better at avoiding spurious signals.

\myparab{Predicting user popularity in social networks.}
We collected tweet data on political topics from October 2016 to December 2017.
Our goal is to predict a user's \emph{next 1-day} popularity, which is defined as the sum of retweets, quotes, and replies received by the user. There are a total of 19 million distinct users, and due to the huge size, we extract the subset of 2000 users with the most interactions for evaluation.
For each user in the 2000-user set, we use its past 5 days' popularity as features. We further randomly sample 200 users and make predictions for them, i.e., setting $d_2 = 200$ to make $d_2$ of the same magnitude as $n$. 

\mypara{Results.}
We randomly sample users for 10 times and report the average MSE and correlation (with standard deviations) for both \emph{in-sample} and \emph{out-of-sample} data. In Table~\ref{table:main_tweets} we can see  results consistent with the equity returns experiment: \emph{(i)}
$\proc{Adaptive-RRR}$ yields the best performance in out-of-sample MSE and correlation. \emph{(ii)}
$\proc{Adaptive-RRR}$ achieves the best generalization error by having a much smaller gap between training and test metrics.

\subsection{Setup of experiments}\label{asec:exp}
\subsubsection{Equity returns}
We use daily historical stock prices and volumes from an emerging market to build our model. Our license agreement prohibits us to redistribute the data so we include only a subset of samples. Full datasets can be purchased by standard vendors such as quandl or algoseek. Our dataset consists of approximately 3,600 stocks between 2011 and 2018. 

\myparab{Universes.} We examine two different universes in this market \emph{(i) Universe 1} is equivalent to S\&P 500. It consists of 800 stocks at any moment. Similar to S\&P 500, the list of stocks appeared in universe 1 is updated every 6 months. A total number of 983 stocks have appeared in universe 1 at least once. \emph{(ii) Universe 2} consists of all stocks except for illiquid ones. This excludes the smallest 5\% stocks in capital and the smallest 5\% stocks in trading volume. Standard procedure is used to avoid universe look-ahead and survival bias~\cite{zhou2014active}.  

\myparab{Returns.} We use return information to produce both features and responses. Our returns are the ``log-transform'' of all open-to-open returns~\cite{zhou2014active}. For example, the next 5-day return of stock $i$ is $\log(p_{i, t+5}/p_{i,t})$, where $p_{i,t}$ is the open price for stock $i$ on trading day $t$. Note that all non-trading days need to be removed from the time series $p_{i,t}$. Similarly, the past 1-day return is $\log(p_{i,t}/p_{i,t-1})$.

\myparab{Model.} We focus on predicting the \emph{next 5-day returns} for different universes. 
Let $\mr_t = (r_{1, t}, r_{2, t}, \dots, r_{d_2, t})$, where $r_{i,t}$ is the next 5-day return of stock $i$ on day $t$.  
Our regression model is 
\begin{align}
\mr_{t+1} = M \mx_t + \epsilon. 
\end{align}

The features consist of the past 1-day, past 5-day, and past 10-day returns of all the stocks in the same universe. For example, in \emph{Universe 1}, the number of responses is $d_2 = 800$. The number of features is $d_1 = 800 \times 3 = 2,400$. We further optimize the hyperparameters $k_1$ and $k_2$ by using a validation set because our theoretical results are asymptotic ones (with unoptimized constants). Baseline models use the same set of features and the same hyper-parameter tuning procedure.

We use three years of data for training, one year for validation, and one year for testing. 
The model is re-trained every test year. 
For example, the first training period is May 1, 2011 to May 1, 2014. The corresponding validation period is from June 1, 2014 to June 1, 2015.
We use the validation set to determine the hyperparameters and build the model, and then we use the trained model to forecast returns of equity in the same universe from July 1, 2015 to July 1, 2016. Then the model is retrained by using data in the second training period (May 1, 2012 to May 1, 2015). This workflow repeats. To avoid looking-ahead, there is a gap of one month between training and validation periods, and between validation and test periods. 

We use standard approach to translate forecasts into positions~\cite{grinold2000active,prigent2007portfolio,ang2014asset,zhou2014active}. Roughly speaking, the position is proportional to the product of forecasts and a function of average dollar volume. We allow short-selling. We do not consider transaction cost and market impact. We use Newey-West estimator to produce $t$-statistics of our forecasts. 
\subsubsection{User popularity}
We use Twitter dataset to build models for predicting a user's \emph{next 1-day} popularity, which is defined as the sum of retweets, quotes, and replies received by the user.

\myparab{Data collection.}
We collected 15 months Twitter data from October 01, 2016 to December 31, 2017, using the Twitter streaming API. We tracked the tweets with topics related to politics with keywords ``trump”, ``clinton”, ``kaine”, ``pence”, and ``election2016”. There are a total of 804 million tweets and 19 million distinct users. User $u$ has one interaction if and only if he or she is retweeted/replied/quoted by another user $v$. Due to the huge size, we extract the subset of 2000 users with the most interactions for evaluation.

\myparab{Model.} Our goal is to forecast the popularity of a random subset of 200 users.
Let $\my_t = (y_{1,t}, \dots, y_{d_2, t})$, where $d_2 = 200$ and $y_{i,t}$ is the popularity of user $i$ at time $t$. Our regression model is 
\begin{equation}
    \my_{t+1} = M\mx_t + \epsilon.
\end{equation}

\myparab{Features.} For each user, we compute his/her daily popularity for 5 days prior to day $t$. Therefore, the total number of features is $d_1 = 2000 \times 5 = 10,000$. 

We remark that there are $n = 240$ observations. This setup follows our assumption $d_1 \gg d_2 \approx n$. 

\myparab{Training and hyper-parameters.}
We use the period from October 01, 2016 to June 30, 2017 as the training dataset, the period from July 01, 2017 to October 30 as validation dataset to optimize the hyper-parameters, and the rest of the period, from September 10, 2017 to December 31, 2017, is used for the performance evaluation.

\fi
\section{Conclusion}
\vspace{-0.3cm}
This paper examines the low-rank regression problem under the high-dimensional setting. We design the first learning algorithm with provable statistical guarantees under a mild condition on the features' covariance matrix. Our algorithm is simple and computationally more efficient than low rank methods based on optimizing nuclear norms. Our theoretical analysis of the upper bound and lower bound can be of independent interest. Our preliminary experimental results demonstrate the efficacy of our algorithm. The full version explains why our (algorithm) result is unlikely to be known or trivial. 

\section*{Broader Impact}
The main contribution of this work is theoretical. Productionizing downstream applications stated in the paper may need to take six months or more so there is no immediate societal impact from this project.

\section*{Acknowledgement}
We thank anonymous reviewers for helpful comments and suggestions. Varun Kanade is supported in part by the Alan Turing Institute under the EPSRC grant EP/N510129/1.
Yanhua Li was supported in part by NSF grants IIS-1942680 (CAREER), CNS-1952085, CMMI-1831140, and DGE-2021871.
Qiong Wu and Zhenming Liu are supported by NSF grants NSF-2008557, NSF-1835821, and NSF-1755769.
The authors acknowledge William \& Mary Research Computing for providing computational resources and technical support that have contributed to the results reported within this paper.

 
\bibliographystyle{plain}
\bibliography{relate}

\end{document}